\documentclass[journal]{IEEEtran}
\IEEEoverridecommandlockouts
\usepackage{url}            
\usepackage{amsfonts}       
\usepackage{amsthm}
\usepackage{amsmath}
\usepackage{nicefrac}       
\usepackage{microtype}      
\usepackage{lipsum}
\usepackage{fancyhdr}  
\usepackage{graphicx}       
\graphicspath{{media/}}     
\usepackage[shortlabels]{enumitem}
\usepackage{booktabs,array}
\usepackage{multirow}
\usepackage{multicol}
\usepackage{makecell}
\usepackage{mathtools}
\usepackage{cite}
\usepackage{caption}
\usepackage{pgfplots}
\usepackage{pgfplotstable}
\usepackage{algorithm}
\usepackage{algpseudocode}
\usepackage{subcaption}

\usepackage{tikz}
\usetikzlibrary{shapes.geometric, arrows, arrows.meta, calc, fit, positioning, backgrounds, decorations.pathreplacing, decorations.markings}

\usepgfplotslibrary{groupplots}
\pgfplotsset{compat = newest}
\pgfplotsset{compat=1.16}
\usepgfplotslibrary{statistics}

\newcommand*{\ReadOutElement}[4]{%
    \pgfplotstablegetelem{#2}{[index]#3}\of{#1}%
    \let#4\pgfplotsretval
}

\usepackage{balance}
\usepackage{tikz}
\usepackage{makecell}
\usepackage{pgfplots}
\usepackage{pgfplotstable}
\usepgfplotslibrary{groupplots}
\usepackage{tikz}
\usetikzlibrary{shadows} 

\newtheorem{proposition}{Proposition}

\definecolor{darkgreen}{rgb}{0.0,0.5,0.0}

\def\BibTeX{{\rm B\kern-.05em{\sc i\kern-.025em b}\kern-.08em
    T\kern-.1667em\lower.7ex\hbox{E}\kern-.125emX}}
\begin{document}

\title{A mechanism for TCP-aware Active Queue Management in packet-forwarding devices\\
}

\author{
\IEEEauthorblockN{Swarnarup Bhunia\IEEEauthorrefmark{1}, Somit Gond\IEEEauthorrefmark{1}, Vayun Goel, Sreelakshmi Manjunath\IEEEauthorrefmark{2} and Padmanabhan Rajan}
\IEEEcompsocitemizethanks{\IEEEcompsocthanksitem\IEEEauthorrefmark{1}Equal contribution \IEEEauthorrefmark{2}Corresponding author}\\
\IEEEauthorblockA{\textit{School of Computing and Electrical Engineering}, 
\textit{Indian Institute of Technology Mandi, HP, India}\\
\texttt{\{swarnarupbhunia2003, somitgond2\}@gmail.com, \{sreelakshmi, padman\}@iitmandi.ac.in}}
}


\maketitle

\begin{abstract}
The deployment of Active Queue Management (AQM) policies is known to be a promising solution to bufferbloat, the problem of persistently-full large buffers in Internet routers. However, most AQMs involve tunable parameters, and some require awareness of the characteristics of the transport-layer protocol implemented at the end systems, thereby rendering their deployment infeasible. We consider one such AQM\textemdash a threshold-based queue policy which requires the packet-dropping threshold to be co-designed with relevant TCP parameters. We propose a mechanism that detects network congestion, estimates the congestion window evolution in the end hosts, estimates the required TCP parameter, and produces an optimal value of the packet-dropping threshold. Arguments based on analysis of the TCP fluid model are presented to substantiate the proposed mechanism. A simulation-based performance evaluation reveals the effectiveness of the proposed mechanism, compared to some widely implemented queue policies, in terms of various user- and system-oriented metrics. The proposed mechanism is particularly well-suited for implementation in Software Defined Networks (SDN) where the algorithms can run in the SDN controller and the packet-dropping threshold can be implemented in the data-plane switches through open interfaces. 
\end{abstract}

\begin{IEEEkeywords}
Bufferbloat, Active Queue Management, Latency, Congestion Control
\end{IEEEkeywords}

\let\thefootnote\relax\footnotetext{A part of this work appears as Bhunia, Swarnarup, et al., ``A mechanism for TCP-aware adaptive management of queues in Internet routers.'' in Proceedings of the 18th International Conference on COMmunication Systems and NETworks (COMSNETS). IEEE, 2026.}

\section{Introduction}
\label{sec:Introduction}

Oversized buffers in packet-forwarding devices, such as routers and switches, are known to cause large queues that lead to increased latency across the Internet~\cite{McFillin2022}. This phenomenon, termed \emph{bufferbloat}~\cite{gettys2012bufferbloat}, is attributed to the use of the simple DropTail policy in these queues and is known to be the cause of congestion collapse in the nascent stages of the Internet~\cite{allman2012comments,cerf2014bufferbloat}. The use of effective Active Queue Management (AQM) strategies has been identified as a solution~\cite{gettys2012bufferbloat,nichols2012controlling, harrison2023buffer}. Although numerous AQM policies have been proposed; see the survey in~\cite{adams2012active}, most of them require precise parameter tuning, making them difficult to deploy in edge or core routers. Consequently, AQMs are not actively deployed in today's networks. This prompted the Internet Engineering Task Force to update the recommendations for the selection and deployment of AQMs, see RFC 7567~\cite{baker2015ietf}. The most important, arguably, of these recommendations was that an AQM must be able to self-tune its parameters on the basis of operational conditions. 

In the recent past, machine learning (ML) approaches have been adopted for AQM design, leading to AQMs which can self-tune their parameters. 
For a detailed survey, see~\cite{toopchinezhad2025machine}. Most of these AQMs, for example, LearnQueue~\cite{bouacida2018practical}, RLECN~\cite{jung2023rlecn}, PET~\cite{cheng2024pet} etc., use reinforcement learning to tune buffer size or packet-marking thresholds. Implementation of such AQMs requires large amounts of training data and their operation incurs significant computational overhead~\cite{toopchinezhad2025machine}. Moreover, the self correcting mechanisms inherent to reinforcement learning may cause instabilities in the network leading to increased jitter and unfairness.
Consequently, one cannot definitively conclude that these ML-based AQMs can be deployed in real-world networks or that they outperform traditional ones~\cite{toopchinezhad2025machine,jiang2021machine}.
%

Parallelly, the emergence of Software Defined Network (SDN) architectures has enabled the implementation of the control plane of the network layer in remote servers. With this, queue-management decisions at packet-forwarding devices can now rely on algorithms implemented in remote servers, thereby leveraging their computational capabilities. In particular, the SDN controller can implement algorithms that enable self-tuning of AQM policies deployed in SDN switches. However, despite this capability, existing research on AQM for SDNs has largely focused on developing enhancements of classical AQM schemes that are amenable to implementation in programmable SDN switches. 
For example, P4-CoDel~\cite{kundel2018p4} is based on CoDel, iRED~\cite{de2026ired} on Random Early Detection (RED), and PI2~\cite{de2016pi2} is an enhancement of the Proportional Integral (PI) queue policy. 
Although implementable, these enhanced AQMs carry the demerits of the underlying core techniques. For instance, iRED uses an exponential weighted moving average (EWMA) of the queue size to make packet-dropping decisions~\cite{de2026ired}, similar to the original RED policy, which is known to make networks prone to instability and their operation highly sensitive to the choice of EWMA parameters and packet-dropping thresholds~\cite{la2004instability}. Consequently, practical deployment of AQMs remains challenging. 

Studies reveal that to effectively deploy most AQMs, their parameters and packet-dropping/marking thresholds need to be tuned in accordance with TCP parameters, inherent to the congestion control algorithm in the transport layer of the end hosts, which are unknown to the router/switch deploying the AQM~\cite{hollot2002analysis,manjunath2018stability,manjunath2019compound,raina2015stability}. 
Leveraging this, we develop a mechanism that can learn the congestion-control behaviour of the transport protocol implemented in the end hosts and make optimal packet-dropping decisions at the routers/switches, thereby enabling an AQM to self-tune its parameters for prevalent operational conditions. In this paper, we focus on a recently proposed threshold-based queue policy~\cite{manjunath2018stability,manjunath2019compound}, which requires that its packet-dropping threshold be co-designed with relevant TCP parameters. We outline a mechanism that estimates the congestion window maintained by the end hosts and estimates the required TCP parameter(s) from the estimated congestion window, and then computes the packet-dropping threshold based on the estimate. The so-computed threshold is used to activate the AQM once the router/switch experiences congestion. The proposed mechanism uses (i) a simple supervised learning technique to estimate the congestion window, and (ii) fundamental principles of TCP congestion control to estimate the relevant parameters. The use of supervised learning lowers the training requirement, compared to other reinforcement-learning-based AQMs which typically require rigorous training and precise tuning. Furthermore, our approach is grounded in mathematical understanding of the underlying TCP-AQM feedback system, which confirms that maintaining control-theoretic stability of the feedback system guarantees improved network performance. 
Therefore, the proposed algorithms promise stable operation of the network at an operating point that ensures improved performance. 
The performance of the proposed mechanism is evaluated using a simulation-based study conducted on ns-3~\cite{riley2010ns}. It is compared with the widely used DropTail policy and the Adaptive RED (ARED) policy which is currently implemented\textemdash although not deployed\textemdash in most IP routers and SDN switches. The mechanism, supported by analytical arguments and inferences from the simulations, leads to TCP-aware queue management for low-latency operation of packet-forwarding devices. It is particularly well suited to SDN architectures where all the computations required for the learning-based decision making are performed in the SDN control server and the decisions are installed in the switches through P4 programming~\cite{bosshart2014p4}. This makes the algorithms implementable in current \textit{light-weight} SDN switches. 

The key contributions of this paper are as follows.
\begin{enumerate}[(i)]
    \item An algorithm for congestion detection through observations of queue size at a packet-forwarding device. 
    \item Application of a simple supervised-learning technique to estimate the congestion window maintained at the end hosts from sequence numbers from the IP datagram payload. 
    \item Methodology for estimating a representative TCP parameter from the estimated congestion window traces supported by relevant mathematical derivations. 
    \item An algorithm for computing the packet-dropping threshold, using the estimated TCP parameter in conjunction with an established control-theoretic result, for stable operation of the network with improved network performance.
    \item Simulation-based performance evaluation with relevant topologies, traffic compositions and metrics, and comparison with relevant baseline policies.
\end{enumerate}
The rest of the paper is organised as follows. The preliminaries and the problem formulation are presented in Section~\ref{sec:Preliminaries}. Section~\ref{sec:ProposedMechanism} outlines the proposed mechanism, and Section~\ref{sec:PerformanceEvaluation} presents the simulations and results. We conclude in Section~\ref{sec:Conclusions}.

\section{Preliminaries}
\label{sec:Preliminaries}
Consider a single bottleneck link shared by $N$ long-lived TCP flows. Treating the evolution of the congestion window maintained by a sending host as a Poisson counter driven stochastic differential equation yields a fluid model for the evolution of the expectation of the congestion window~\cite{misra2000fluid}. When the Poisson processes pertaining to the $N$ end hosts are assumed to be independent and identically distributed, the aforementioned fluid model captures the \emph{average} congestion window size of the $N$ TCP flows~\cite{misra2000fluid}. Denote this average congestion window at time $t$ as $w(t)$.
Suppose that the end hosts deploy a TCP variant that follows Additive Increase Multiplicative Decrease (AIMD) congestion control, where it increases the sending window by a general additive increase (AI) function, say $i\big(w(t)\big)$, per acknowledgement received and reduces it by a multiplicative decrease (MD) $d\big(w(t)\big)$ per packet drop, the fluid model for the evolution of the average congestion window $w(t)$ is updated as~\cite{raina2005part}
\begin{align}
 \dot{w}(t) = \bigg(i\big(w(t)\big)\big(1-p(t-\tau)\big)-d\big(w(t)\big)p(t-\tau)\bigg)\frac{w(t-\tau)}{\tau},\label{eq:gen_TCP}
\end{align}
where $p(t)$ represents the packet-drop probability and $\tau$ is the average round-trip time (RTT) of the TCP flows. In the treatment of TCP fluid models, it is assumed that packets are dropped, if at all, only at the bottleneck link. 
Now, consider the multiplicative decrease component alone. The decrease function can be written as $d(w) = \beta w(t),$ where  $\beta \in (0,1)$ is the MD parameter. 
In particular, when a packet is lost, the average congestion window is updated as~\cite{raina2005part, manjunath2018stability} 
    \begin{align}
        w(t+1) = (1-\beta)w(t). \label{eq:avg_cwnd_betaestimate} 
    \end{align}
Note that, the congestion window update inherently depends on the packet-drop probability as TCP infers network congestion using the acknowledgement received; this probability depends on the packet-dropping decisions made by the AQM in a packet-forwarding device such as a router/switch. 

The threshold-based queue policy proposed in~\cite{manjunath2018stability,manjunath2019compound} makes binary packet-dropping decisions as
\begin{align*}
    \begin{cases}
    \text{if } q(t) \geq q_{th}, & \text{drop packet}\\
    \text{else}, & \text{enqueue}
    \end{cases}
\end{align*}
where $q(t)$ is the size of the queue, in the router/switch preceding the bottleneck link, at time $t$ and $q_{th}$ is a tunable parameter of the queue policy that serves as a packet-dropping threshold. When the packet-drop probability $p(t)$ is modelled based on the binary decisions made by the threshold-based queue policy, we obtain a closed-loop fluid model for a set of TCP flows feeding into a bottleneck link managed by this queue policy. A control-theoretic analysis of the fluid model reveals that the threshold-based queue policy can maintain network stability and decongest the bottleneck link, while ensuring reduced queueing delays, when $q_{th}$ is tuned according to the sufficient condition~\cite{manjunath2018stability}
\begin{align}
    q_{th}\,\beta\, w^\ast\, \left(\frac{w^\ast}{\tilde{C}\tau}\right)^{q_{th}} \leq \pi/2.\label{eq:sufficient_condition}
\end{align}
Here, $w^\ast$ is the steady-state average congestion window of the TCP flows, which is the fixed point of the fluid model~\eqref{eq:gen_TCP} obtained using $\dot{w}(t) = 0$~\cite{manjunath2018stability}. Further, $\tilde{C}$ is the per-flow capacity of the bottleneck link. 

To implement the threshold-based queue policy and reap its benefits, a router/switch must be able to compute the packet-dropping threshold $q_{th}$ that satisfies condition~\eqref{eq:sufficient_condition}. This computation requires knowledge of $\tau$, $w^\ast$ and $\beta$, assuming that the capacity of the output link $\tilde{C}$ is known to the router. Using ACK-clocking methods, the router can estimate the average RTT of the TCP flows, $\tau$. In order to estimate $w^\ast$, the router/switch needs to estimate the evolution of the congestion window maintained by the end hosts.
Further, the estimated congestion windows of the hosts can be used to estimate the parameter $\beta$ through equation~\eqref{eq:avg_cwnd_betaestimate}. 

Recall that the aforementioned TCP fluid model and the analysis rely on the assumption that all hosts use the same AIMD-based congestion control algorithm in the default parameter settings. 
Despite this similarity, the empirical average of the congestion windows of TCP flows sharing a bottleneck link does not strictly follow equation~\eqref{eq:avg_cwnd_betaestimate}; a more elaborate discussion on this is deferred to Section~\ref{sec:parameter_estimation}. 
In such a scenario, one must estimate a representative $\beta$ parameter that governs the multiplicative decrease in the average congestion window. With this estimate at hand, inequality~\eqref{eq:sufficient_condition} can be used to compute a value of the threshold $q_{th}$ required to deploy the threshold-based policy.  



\begin{figure}
    \centering
        \begin{subfigure}{\linewidth}
    \centering
    \begin{tikzpicture}[titlerect/.style={draw,inner ysep=0.5ex,inner xsep=0.5ex,  align=center,label={[anchor=center,fill=white]above:#1}}]
    
    \node[rectangle, rounded corners, text width= 2.05 cm, inner ysep=0.5ex,inner xsep=0.4ex, draw] (CD) at (0,0) {\begin{small}
     \hspace{-1mm}(A) Congestion detection\\
     \textbf{Algorithm 1}
    \end{small}
    };

    \node[rectangle, rounded corners, text width=2.05 cm, below = 1cm of CD, inner ysep=0.5ex, inner xsep=0.4ex, draw] (CWP)
    {\begin{small}
    \hspace{-1mm}(B) Congestion window estimation\\
    \textbf{Random forest regressor}
    \end{small}
    };

    \node[diamond, rounded corners, right = 0.8cm of CD,  draw, inner ysep=0ex,inner xsep=0ex] (CDQ) {\begin{small}
    Congestion?
    \end{small}
    };

    \node[rectangle, rounded corners, inner ysep=0.5ex,inner xsep=0.6ex, left=0.5cm of CD, text width=1cm, yshift=-1.25cm] (QnNI) {\begin{small}{Inputs from switch}\end{small}};

    
    \node[rectangle, rounded corners, text width= 2.25 cm,below =0.7cm of CDQ,inner ysep=0.5ex,inner xsep=0.4ex,  draw] (BE) {\begin{small}
     \hspace{-1mm}(C) $\beta$ estimation
    \end{small}
    };
    \node[rectangle, rounded corners, text width= 2.25 cm,below = 0.5cm of BE,inner ysep=0.5ex,inner xsep=0.4ex, draw] (QT) {\begin{small}
     \hspace{-1mm}(D) Threshold $q_{th}$ computation\\
     \textbf{Algorithm 2}
    \end{small}
    };
    \node[rectangle, rounded corners,text width=0.3cm, inner ysep=0ex,inner xsep=0.1ex, right = 0.65cm of QT] (Dep) {\begin{small}{$q_{th}$}\end{small}};

    \node[titlerect={\begin{small}Proposed Mechanism\end{small}}, draw,inner ysep=7pt, inner xsep=7pt, dashed,fit={(CD) (BE) (CDQ) (QT)}] (PM) {};
    
    \draw[-triangle 60] ([yshift=0.65cm]CWP.east) -- node[below] {\scriptsize{w$_i(t),$}} (BE.west); 
    
    \draw[-triangle 60] ([yshift=-0.4cm]CWP.east) -- node[above,yshift=0.25cm] {\hspace{5mm}\scriptsize{$ i \in [1,N]$}} ([yshift=0.3cm]QT.west); 
    
    \draw[-triangle 60] (CD.east) -- (CDQ.west);
    
    \draw[-triangle 60] (CDQ.south) -- node[right] {\scriptsize{Yes}} (BE.north);
    
    \draw[-triangle 60] (BE.south) -- (QT.north);
    
    \draw[-triangle 60] (QT.east) -- (Dep.west);

    \draw[-triangle 60] (QnNI.north) |- node[yshift=0.275cm,xshift=-0.1cm,above] {\small{Queue}} (CD.west);
    \draw[-triangle 60] (QnNI.north) |- node[above] {\small{size, $q(t)$}} (CD.west);
    
    \draw[-triangle 60] (QnNI.south) |- node[below] {\small{IP Adds}} (CWP.west);
    \draw[-triangle 60] (QnNI.south) |- node[yshift=-0.3cm,xshift=0.05cm,below] {\small{Port nos.}}  (CWP.west);

    
    \end{tikzpicture}
    \caption{Schematic representation of the proposed mechanism.}
    \label{fig:schematic_PM}
    \end{subfigure}\\
    \vspace{3mm}
    \begin{subfigure}{\linewidth}
    \centering
    \begin{tikzpicture}[titlerect/.style={draw,inner ysep=0.5ex,inner xsep=0.5ex,  align=center,label={[anchor=center,fill=white]above:#1}}]
    \node[trapezium,
        trapezium left angle=70,
        trapezium right angle=110, text width= 3.5 cm, inner ysep=0.25ex,inner xsep=0.2ex, draw, drop shadow] (SDN) at (0,0) {\begin{center}\begin{small}
    SDN control server\\
    Implements (A, B, C, D)
    \end{small}\end{center}
    };
    \node[rectangle, below =2cm of SDN, minimum width=1.75cm,minimum height = 0.35cm, draw,xshift=-0.25cm] (Que) {};
    
    \node[rectangle, rounded corners,text width=0.3cm, inner ysep=0.15ex,inner xsep=0.1ex, left=1.25cm of Que,yshift=0.3cm] (IN) {\scriptsize{Ingress Pkts}};
    
    \foreach \i in {0,1,2} {
    \path let \p1 = (IN) in
    node[fill=blue!40,draw=blue!80,
           minimum width=0.1cm, minimum height=0.25cm]
      (pkt\i) at ($(\p1)+(0.5+0.35*\i,-0.3)$) {};
    }

    \foreach \j in {0,1,2,3,4} {
    \path let \p2 = (IN) in
    node[fill=blue!40,draw=blue!80,
           minimum width=0.1cm, minimum height=0.25cm]
      (pkt_que\j) at ($(\p2)+(1.95+0.27*\j,-0.3)$) {};
    }

    \node[fill=blue!40,draw=blue!80,
           minimum width=0.1cm, minimum height=0.25cm,right=0.12cm of pkt_que4]
      (pkt_in) {};

    \foreach \k in {0,1,2} {
    \path let \p3 = (Que) in
    node[fill=red!40,draw=red!80,
           minimum width=0.1cm, minimum height=0.25cm]
      (pkt_drop\k) at ($(\p3)+(-0.5, -1.2+0.35*\k)$) {};
      }

    \node[circle, right =0.025cm of Que, minimum width=0.4cm, draw] (QueSer) {};
    \node[below=0.05cm of Que,xshift=0.425cm] (Buffer){\scriptsize{Buffer}};

    \node[rectangle, rounded corners,text width=0.3cm, inner ysep=0.2ex,inner xsep=0.1ex, above=0.2cm of Que,xshift=-0.05cm] (Sw) {\begin{small}Switch\end{small}};
    
    \node[rectangle, draw,inner ysep=2pt, inner xsep=7pt, solid, fit={(Que) (QueSer) (Sw) (Buffer)}] (Switch) {};

    \draw[-triangle 60] ([xshift=-0.2cm]Switch.north) -- node[left,yshift=0.125cm] {\small{Queue size}} node[left,yshift=-0.15cm]{\small{IP Adds.}} node[left,yshift=-0.45cm]{\small{ Port nos.}}([xshift=-0.25cm]SDN.south);
    \draw[-triangle 60] ([xshift=0.25cm]SDN.south) --  node[right] {\small{Threshold $q_{th}$}}([xshift=0.275cm]Switch.north);

    \node[rectangle, rounded corners,text width=0.3cm, inner ysep=0.1ex,inner xsep=0ex, right=1.25cm of QueSer] (For) {\scriptsize{Forward}};

    \draw[-triangle 60] ([xshift=0.25cm]QueSer.east) -- (For.west);
    \end{tikzpicture}
      \caption{Implementation framework for the proposed mechanism.}
      \label{fig:schematic_implementation}
    \end{subfigure}
    \caption{Pictorial representation showing the proposed mechanism and its implementation.}
    \label{fig:blockDiagram}
\end{figure}
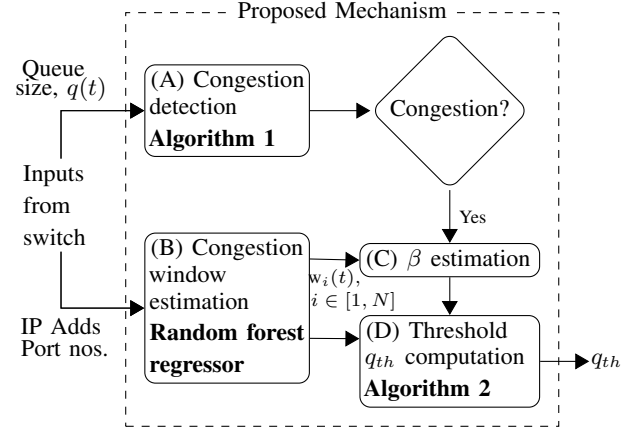
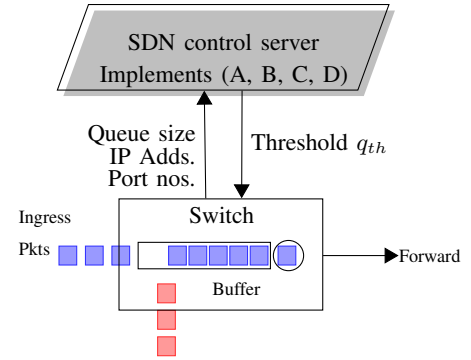

\subsection{Problem formulation}
Deploying the threshold-based queue policy in a packet-forwarding device boils down to implementing algorithms for the following.
\begin{enumerate}[(i)]
    \item Detect congestion at the router/switch by observing the queue size. 
    \item Estimate the congestion window evolution for all end hosts sharing the bottleneck link.
    \item Estimate the multiplicative decrease parameter $\beta$ that governs the decrease in the average congestion window after packet loss in the event of congestion. 
    \item Compute an appropriate value of the packet-dropping threshold $q_{th}$ which decongests the router/switch and maintains low queueing delays. 
\end{enumerate}
For developing the underlying algorithms for the mechanism, we assume that all the end systems use AIMD-based congestion control and all the flows are long-lived, \emph{i.e.}, flows pertaining to FTP applications that last long enough to drive the system into a steady state. This simplistic choice of traffic composition aligns with the assumptions of the fluid models employed in the study of TCP-AQM feedback systems and ensures the applicability of the analytical result~\eqref{eq:sufficient_condition}. Once we develop the required algorithms,  we systematically deviate from the assumptions and evaluate their performance in more evolved network topologies and realistic traffic mixes through simulations. We now outline the required algorithms.
\section{Proposed Mechanism}
\label{sec:ProposedMechanism}
The proposed mechanism is divided into four parts, each of which is described below. The interplay between these is shown in Fig.~\ref{fig:schematic_PM}, and the possible implementation is depicted in Fig.~\ref{fig:schematic_implementation}. 
\subsection{{Congestion detection}}
\label{sec:CongestionDetection}
It is well known that congestion causes synchronisation of TCP flows across hosts sharing a link, and results in pronounced deterministic oscillations in the queue size~\cite{han2005synchronization,ghosh2020right,raina2005buffer}. Building on this premise, we use queue size oscillations as an indicator of incipient congestion. Such oscillations can be detected and distinguished from other variations in the queue size by computing the autocorrelation of the queue size data. Oscillations in the queue size with relatively higher amplitude and lower frequency, which indicate the onset of congestion, cause the autocorrelation of the queue size to have lesser number of zero crossings. This relationship is illustrated in Fig.~\ref{fig:CongestionDetection}, where we plotted the queue size at the router observed through packet-level simulations of a single-bottleneck link of 100 Mbps capacity serving $N=60$ TCP hosts. The conventional DropTail policy is used for queue management. We used the TCP NewReno variant and simulated two regimes, one where the TCP flows have a small average RTT of 10 ms, and another where the average RTT is 200 ms. A large bandwidth-delay product, or in particular large RTTs, is known to cause congestion in networks that use DropTail~\cite{raina2005buffer}; indeed, we see queue size oscillations in this case with large amplitude and low frequency as shown in Fig.~\ref{fig:CongestionDetection}(b). We computed the autocorrelation of the queue size, considering all possible lags. We observed that the autocorrelation exhibits a lesser number of zero crossings for the large RTT case. 
Therefore, observing the number of zero crossings of the autocorrelation of the queue size can help detect congestion. Algorithm~\ref{alg:CongestionDetection} outlines congestion detection using this mechanism. 


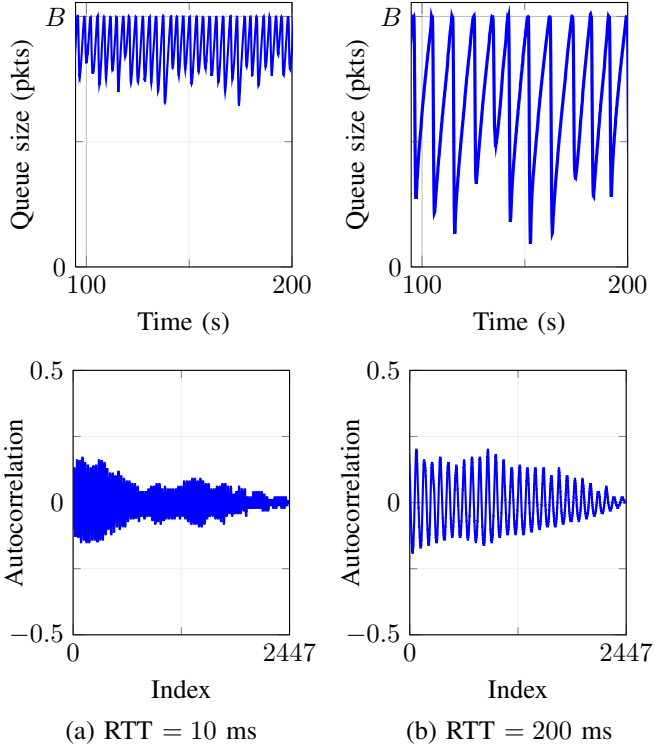
\begin{figure}[t]
	\centering
    \begin{tabular}{c c}
	\begin{tikzpicture}
		\begin{axis}[
		yticklabel style={
			/pgf/number format/fixed,
			/pgf/number format/precision=2,
			/pgf/number format/fixed zerofill,
			/pgf/number format/fixed relative,
		},
		every axis plot/.append style={thick},
        every tick label/.append style={font=\normalsize},
        every axis label/.append style={font=\normalsize},
        every axis y label/.append style={yshift=-3pt},
		ylabel = {Queue size (pkts)},
		xlabel = {Time (s)},
		xmin = 100, xmax = 200,
		ymin = 0, ymax = 2200, 
        xtick={105,200},
        xticklabels={$100$,$200$},
		ytick={0,2084},
        yticklabels={$0$,$B$},
		grid = both,
		minor tick num = 1,
		major grid style = {lightgray},
		minor grid style = {lightgray!25},
		width=1.75in,
        height=2in,
		legend cell align = {right},
		legend pos = south west
		]
        \addplot[blue]  table [x = {t}, y ={q}] {Data/tc-qsizeTrace-dumbbell_rtt15ms.txt};
		\end{axis}
		\end{tikzpicture}
      & \hspace{-0.45cm}
	\begin{tikzpicture}[thick]
		\begin{axis}[
		yticklabel style={
			/pgf/number format/fixed,
			/pgf/number format/precision=2,
			/pgf/number format/fixed zerofill,
			/pgf/number format/fixed relative,
		},
		every axis plot/.append style={thick},
        every tick label/.append style={font=\normalsize},
        every axis label/.append style={font=\normalsize},
        every axis y label/.append style={yshift=-3pt},
		ylabel = {Queue size (pkts)},
		xlabel = {Time (s)},
		xmin = 100, xmax = 200,
		ymin = 0, ymax = 2200, 
        xtick={105,200},
        xticklabels={$100$,$200$},
		ytick={0,2084},
        yticklabels={$0$,$B$},
		grid = both,
		minor tick num = 1,
		major grid style = {lightgray},
		minor grid style = {lightgray!25},
		width=1.75in,
        height=2in,
		legend cell align = {right},
		legend pos = south west
		]
        \addplot[blue,style=very thick]  table [x = {t}, y ={q}] {Data/tc-qsizeTrace-dumbbell_rtt200ms.txt};
		\end{axis}
		\end{tikzpicture}
        \\
        \hspace{-0.15cm}
	\begin{tikzpicture}
		\begin{axis}[
		yticklabel style={
			/pgf/number format/fixed,
			/pgf/number format/precision=2,
			/pgf/number format/fixed zerofill,
			/pgf/number format/fixed relative,
		},
		every axis plot/.append style={thick},
        every tick label/.append style={font=\normalsize},
        every axis label/.append style={font=\normalsize},
        every axis y label/.append style={yshift=-12pt},
		ylabel = {Autocorrelation},
		xlabel = {Index},
		xmin = 2447, xmax = 4897,
		ymin = -0.5, ymax = 0.5, 
        xtick={2447,4897},
        xticklabels={$0$,$2447$},
		ytick={-0.5,0,0.5},
        yticklabels={$-0.5$, $0$,$0.5$},
		grid = both,
		minor tick num = 1,
		major grid style = {lightgray},
		minor grid style = {lightgray!25},
		width=1.75in,
        height=2in,
		legend cell align = {right},
		legend pos = south west
		]
        \addplot[blue]  table [x expr=\coordindex, y ={a}] {Data/autocorr_rtt15.txt};
		\end{axis}
		\end{tikzpicture}
      & 
      \hspace{-0.45cm}
	\begin{tikzpicture}
		\begin{axis}[
		yticklabel style={
			/pgf/number format/fixed,
			/pgf/number format/precision=2,
			/pgf/number format/fixed zerofill,
			/pgf/number format/fixed relative,
		},
		every axis plot/.append style={thick},
        every tick label/.append style={font=\normalsize},
        every axis label/.append style={font=\normalsize},
        every axis y label/.append style={yshift=-12pt},
		ylabel = {Autocorrelation},
		xlabel = {Index},
		xmin = 2447, xmax = 4897,
		ymin = -0.5, ymax = 0.5, 
        xtick={2447,4897},
        xticklabels={$0$,$2447$},
		ytick={-0.5,0,0.5},
        yticklabels={$-0.5$, $0$,$0.5$},
		grid = both,
		minor tick num = 1,
		major grid style = {lightgray},
		minor grid style = {lightgray!25},
		width=1.75in,
        height=2in,
		legend cell align = {right},
		legend pos = south west
		]
        \addplot[blue]  table [x expr=\coordindex, y ={a}] {Data/autocorr_rtt200.txt};
		\end{axis}
		\end{tikzpicture}\\
        (a) RTT $=10$ ms & (b) RTT $=200$ ms 
        \end{tabular}
	\caption{Plots showing queue size vs time (top panel) and autocorrelation of the queue size (bottom panel), for small and large RTT regime. Here, $B$ is the size of buffer. Large RTT implies a large bandwidth-delay product regime, where congestion is known to occur. Observe that the queue size exhibits pronounced oscillations in this case. When the queue size oscillates excessively, its autocorrelation has lesser zero crossings.}
	\label{fig:CongestionDetection}
\end{figure}

\begin{algorithm}
\caption{Congestion detection}\label{alg:CongestionDetection}
\begin{algorithmic}
\State \textbf{Input}: $q(t)$, $\mathcal{T}$, $m_z$ \Comment{Input, Parameters}
\State $t_{curr} \gets 0,\, flag \gets 0,\, temp  = \phi$,
\While{1}
    \State $ t_{nxt} \gets t_{curr} + \mathcal{T}$\\
    \hspace{5mm}wait until $t_{nxt}$
    \State $q_s \gets q[t_{curr}: t_{nxt}]$
    \State $ar \gets$ normalized autocorrelation$(q_s)$
    \State $z \gets$ zero crossings$(ar)$
    \If{ $z < m_z$}
        \State $temp \gets temp \cup {1}$ \Comment{Append 1 to $temp$}
    \Else
        \State $temp \gets temp \cup {0}$ \Comment{Append 0 to $temp$}
    \EndIf
    \State $n \gets \text{length}(temp)$
    \If {$n > 2$ and $\displaystyle{\sum_{i=n-2}^{n}} temp[i] = 3$} 
    \vspace{2mm}\\
    \hfill{\Comment{Check if $z < m_z$ in $3\mathcal{T}$}}
    \State $flag \gets 1$ \Comment{Congestion detected}
    \EndIf
    \State $t_{curr} \gets t_{nxt}$
    \hspace{5mm}\State \textbf{Output:} $flag$
    \State $flag \gets 0$ \Comment{Reset $flag$}
\EndWhile
\end{algorithmic}
\end{algorithm}


\begin{figure}[tbh!]
  \centering
        \begin{tikzpicture}[thick]
        \def\xlabel{False Positive Rate}
        \def\ylabel{True Positive Rate}
        \def\ymin{0.0}
        \def\ymax{1.05}
        \begin{axis}[
            xlabel={\xlabel},
            ylabel={\ylabel},
            xmin=0.0, xmax=1.0,
            ymin=\ymin, ymax=\ymax,
            xtick={0.0,0.2,...,1.0},
            legend pos=south east,
            legend style={draw=none},
            grid=both,
            major grid style={dashed,line width=0.5pt,draw=gray!50},
            minor grid style={dotted,line width=0.25pt,draw=gray!25},
            yticklabel style={/pgf/number format/fixed},
            x tick label style={anchor=north},
            width=3in,
            height=2.5in,
        ]
        
        \addplot[
            color=blue,very thick,
            mark=*,
            mark options={fill=blue}
        ] table[x=FPR,
        y=TPR, 
        col sep=comma,
        restrict expr to domain={\thisrow{winSize}}{5:5}] {Data/roc_analysis_data.csv};
        \addlegendentry{$\mathcal{T}$= 5s}

        \addplot[
            color=red, very thick,
            mark=square*,
            mark options={fill=red},
            dotted
        ] table[x=FPR,
        y=TPR, 
        col sep=comma,
        restrict expr to domain={\thisrow{winSize}}{10:10}] {Data/roc_analysis_data.csv};
        \addlegendentry{$\mathcal{T}$ = 10s}

        \addplot[
            color=darkgreen, very thick,
            mark=triangle*,
            mark options={fill=darkgreen},
            dashed
        ] table[x=FPR,
        y=TPR, 
        col sep=comma,
        restrict expr to domain={\thisrow{winSize}}{15:15}] {Data/roc_analysis_data.csv};
        \addlegendentry{$\mathcal{T}$ = 15s}

        \addplot[
            color=brown,very thick,
            mark=diamond*,
            mark options={fill=brown},
            dash dot
        ] table[x=FPR,
        y=TPR, 
        col sep=comma,
        restrict expr to domain={\thisrow{winSize}}{20:20}] {Data/roc_analysis_data.csv};
        \addlegendentry{$\mathcal{T}$ = 20s}

        \end{axis}
        \end{tikzpicture} 
        \caption{ROC plots showing variation of $\mathcal{T}$ and $m_z$ }
        \label{fig:roc_plots}
\end{figure}
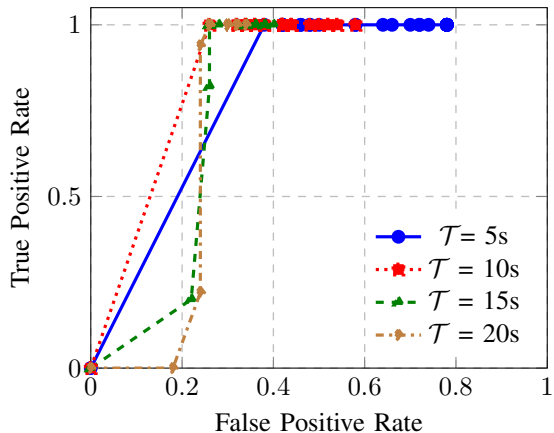

In Algorithm~\ref{alg:CongestionDetection}, $m_z$ is the zero-crossing threshold used to determine the presence of congestion. The value of $m_z$ is decided by the operating point of the algorithm; i.e.~the relative cost of missing the presence of congestion, versus the cost of falsely declaring congestion when there is actually none. Assuming that the cost of false alarms are negligible, the operating point should be selected to minimize misses.
When the number of zero crossings in a time duration of $\mathcal{T}$ is less than $m_z$, then the algorithm concludes that the link is congested. 

The parameter $\mathcal{T}$ represents a time duration for checking the queue size traces for deterministic oscillations. Although the choice of $\mathcal{T}$ is left to the implementation and may vary across routers, it is important to note that $\mathcal{T}$ must be at least a few multiples of RTT. This is because it takes at least a few RTTs for TCP hosts to infer congestion in the network and react to packet loss by reducing their congestion window. Setting $\mathcal{T}$ to a large value, say a few seconds, would ensure that the router detects congestion successfully. However, this may result in a delay in detecting congestion, which may impact network performance. 

Towards understanding the influence of the parameter $\mathcal{T}$ on the ability of Algorithm~\ref{alg:CongestionDetection} to accurately detect congestion, we examined the Receiver Operating Characteristic (ROC) curve of the detector. For this we conducted 100 simulations, of the aforementioned single-bottleneck link, of which 50 were of a congested scenario while the remaining did not have congestion. Taking it as ground truth we compared our congestion detection approach with a range of thresholds $m_z$, for various values of the time duration $\mathcal{T}$. Using this, we plotted the ROC curve to visualize the results. The results are presented in Fig.~\ref{fig:roc_plots}. It can be seen that most values of $\mathcal{T}$ give similar true positive rates (close to 100\%.) To reduce latency, we thus choose to implement Algorithm \ref{alg:CongestionDetection} with $\mathcal{T}$ = 5 seconds.

\subsection{Congestion window estimation}
Whilst detecting congestion, it is also crucial to maintain an accurate estimate of the congestion window of each TCP flow at the router which upon detection of congestion, can help calculate the required parameters to eventually estimate the optimal packet-dropping threshold $q_{th}$ to resolve congestion. This can be achieved by (1) sniffing the sequence number and related fields from the IP datagram payload at the router and then (2) utilising an intelligent algorithm which can estimate the corresponding congestion window size maintained at the TCP sender's end for that corresponding TCP flow. 
Our congestion window estimation methodology involves sniffing the following header fields: the source and destination IP addresses and port numbers, TCP sequence number, acknowledgement number, and the packet size. The four tuple of the IP addresses and port numbers help uniquely identify the TCP flow and these sniffed header fields along with the time of packet arrival can be cached and stored as a time series data per TCP flow. The congestion window size estimation is made using a random forest regressor. The time series data is preprocessed to create the following input features namely direction (whether from source to destination or vice versa), sequence number difference, acknowledgment number difference, time difference (current value minus the previous value), packet size and duplicate ACK count. A window size of 10 packets is considered, i.e. input features of the 9 previous packets along with the currently received input packet features are concatenated and fed into the model as input which in turn estimates the congestion window size at that time instant. These estimates are cached for future estimations in the auto-regressive manner as explained.

The implementation from scikit-learn's \texttt{RandomForestRegressor} model initialized with 100 estimators is used. Training data is generated using ns-3 for each simulation topology considered in the study, the details of which are deferred to Section~\ref{sec:PerformanceEvaluation}. More that 0.8 million training examples were used to train the model. 

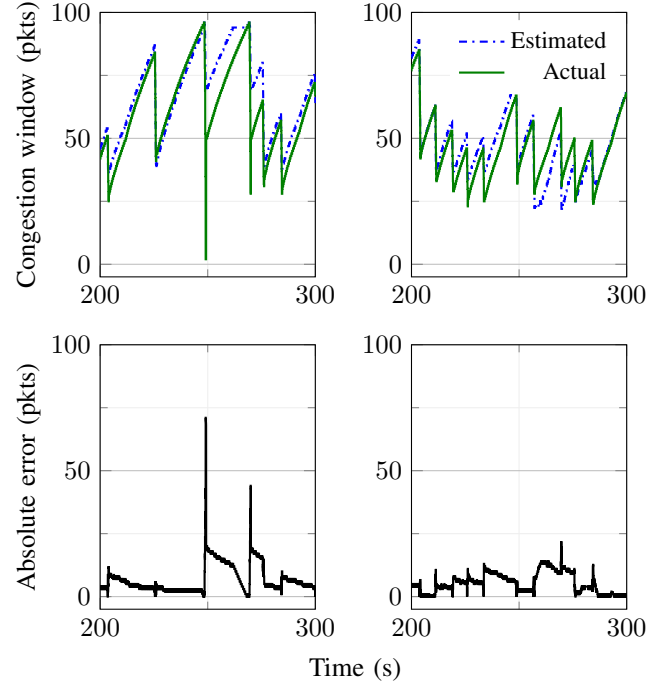
\begin{figure}[t]
	\centering
	\begin{tikzpicture}
		\begin{axis}[
		yticklabel style={
			/pgf/number format/fixed,
			/pgf/number format/precision=2,
			/pgf/number format/fixed zerofill,
			/pgf/number format/fixed relative,
		},
		every axis plot/.append style={thick},
        every tick label/.append style={font=\normalsize},
        every axis label/.append style={font=\normalsize},
        every axis y label/.append style={yshift=-3pt},
		ylabel = {Congestion window (pkts)},
		xmin = 200, xmax = 300,
		ymin = -5, ymax = 100, 
        xtick={200,300},
		ytick={0,50,100},
        yticklabels={$0$,$50$,$100$},
		grid = both,
		minor tick num = 1,
		major grid style = {lightgray},
		minor grid style = {lightgray!25},
        width=0.5\linewidth,
        height=2in,
		legend cell align = {right},
		legend pos = south east
		]
         \addplot[blue,dash dot]  table [col sep=comma,x = T, y = P, each nth point=5] {Data/rf_prediction_data_node_10_200to300_sizeshrinked.csv};
         \addplot[darkgreen,solid] table[col sep=comma, x=T, y=G,each nth point=5]  {Data/rf_prediction_data_node_10_200to300_sizeshrinked.csv};
		\end{axis}
		\end{tikzpicture}
	    	\begin{tikzpicture}
		\begin{axis}[
		yticklabel style={
			/pgf/number format/fixed,
			/pgf/number format/precision=2,
			/pgf/number format/fixed zerofill,
			/pgf/number format/fixed relative,
		},
		every axis plot/.append style={thick},
        every tick label/.append style={font=\normalsize},
        every axis label/.append style={font=\normalsize},
        every axis y label/.append style={yshift=-3pt},
		xmin = 200, xmax = 300,
		  ymin = -5, ymax = 100,
        xtick={200,300},
		ytick={0,50,100},
		grid = both,
		minor tick num = 1,
		major grid style = {lightgray},
		minor grid style = {lightgray!25},
        width=0.5\linewidth,
        height=2in,
		legend cell align = {right},
		legend pos = north east,        
        legend style={draw=none}
		]
         \addplot[blue,dash dot]  table [col sep=comma,x = T, y = P, each nth point=5] {Data/rf_prediction_data_node_40_200to300_sizeshrinked.csv};
         \addlegendentry{\small{Estimated}}
         \addplot[darkgreen,solid] table[col sep=comma, x=T, y=G,each nth point=5]  {Data/rf_prediction_data_node_40_200to300_sizeshrinked.csv};
        \addlegendentry{\small{Actual}}
		\end{axis}
		\end{tikzpicture}\\
        \vspace{2mm}
        	\begin{tikzpicture}
		\begin{axis}[
		yticklabel style={
			/pgf/number format/fixed,
			/pgf/number format/precision=2,
			/pgf/number format/fixed zerofill,
			/pgf/number format/fixed relative,
		},
		every axis plot/.append style={thick},
        every tick label/.append style={font=\normalsize},
        every axis label/.append style={font=\normalsize},
        every axis y label/.append style={yshift=-3pt},
		ylabel = {Absolute error (pkts)},
		xmin = 200, xmax = 300,
		ymin = -5, ymax = 100,
        xtick={200,300},
		ytick={0,50,100},
		grid = both,
		minor tick num = 1,
		major grid style = {lightgray},
		minor grid style = {lightgray!25},
        width=0.5\linewidth,
        height=2in,
		]
         \addplot[black,solid]  table [col sep=comma,x = T, y = E, each nth point=5] {Data/rf_prediction_data_node_10_200to300_sizeshrinked.csv};
		\end{axis}
		\end{tikzpicture}
        \vspace{-4mm}
        	\begin{tikzpicture}
		\begin{axis}[
		yticklabel style={
			/pgf/number format/fixed,
			/pgf/number format/precision=2,
			/pgf/number format/fixed zerofill,
			/pgf/number format/fixed relative,
		},
		every axis plot/.append style={thick},
        every tick label/.append style={font=\normalsize},
        every axis label/.append style={font=\normalsize},
        every axis y label/.append style={yshift=-3pt},
		xlabel = {},
		xmin = 200, xmax = 300,
		ymin = -5, ymax = 100,
        xtick={200,300},
		ytick={0,50,100},
		grid = both,
		minor tick num = 1,
		major grid style = {lightgray},
		minor grid style = {lightgray!25},
        width=0.5\linewidth,
        height=2in,
		legend cell align = {right},
		legend pos = south east
		]
         \addplot[black,solid]  table [col sep=comma,x = T, y = E, each nth point=5] {Data/rf_prediction_data_node_40_200to300_sizeshrinked.csv};
		\end{axis}
		\end{tikzpicture}\\
        \vspace{4mm}
        \centering{\hspace{6mm}Time (s)}
        \caption{Plots showing the estimated congestion window evolution juxtaposed against the actual, for two arbitrarily chosen hosts. The common legend for the comparison plots is placed in the top-right plot alone.}
\end{figure}

\subsection{Parameter estimation}
\label{sec:parameter_estimation}
With the congestion window estimates at hand, we now need to estimate the multiplicative decrease parameter that governs the decrease in the average congestion window from these estimates. We propose a mathematical model for the required estimation. The correctness of the model is established through the propositions presented below. 

When the congestion-control algorithm in TCP source host $i$ infers packet loss, its congestion window $\mathrm{w}_i$ is reduced as
\begin{align}
    \mathrm{w}_i(t+1) = (1-\beta_i)\mathrm{w}_i(t), \label{eq:window_update_individual}
\end{align}
where $\beta_i$ is the MD parameter of the host $i$. From the estimated congestion window of a host $i$, obtained using the random forest regressor, the SDN controller can use $\mathrm{w}_i(t)$ and  $\mathrm{w}_i(t+1)$ to compute $\beta_i$ using equation~\eqref{eq:window_update_individual}. Although this can be done to obtain the MD parameter of each source, these parameters in themselves do not capture the behaviour of the average congestion window $w(t) = (\sum_{i=1}^N \mathrm{w}_i(t))/N$ whose fixed-point value $w^\ast$ appears in inequality~\eqref{eq:sufficient_condition}. The size of the congestion window w$_i(t)$ of each source $i$ is influenced by its AI function $i(w)$. An aggressive source with a higher increase parameter will have a larger congestion window while a conservatively increasing source will have a smaller window. To that end, each source contributes differently to the average congestion window $w(t)$ as seen by the packet-forwarding device. Therefore, the required estimate, that of a parameter that governs the decrease in $w(t)$, is not a straight-forward average of the parameters $\beta_i$. 

Let us define a loss event as one in which the average congestion window is reduced by a fraction due to reduction in the congestion windows of a subset of TCP sources. Suppose that there are $k$ such loss events. Let $s_e$ be the subset of TCP sources which reduced their congestion window in $e^{th}$ loss event, where $e \in \{1,2, \cdots, k\}$. 

\begin{proposition}
    At each loss event $e \in \{1,2,\cdots,k\}$, consider the congestion window averaged over all TCP flows $i \in s_e$, denoted as $w_e(t).$ The decrease in $w_e(t)$ is governed by
    \begin{align}
        w_e(t+1) = (1-\beta_e)w_e(t), \label{eq:avg_cwnd_lossevent}
    \end{align}
    where the parameter $\beta_e$ is a weighted average of the decrease parameters of all the individual TCP sources $i \in s_e$, and is given by 
    \begin{align}
        \beta_e = \displaystyle{\sum_{i \in s_e}}\beta_i \mathrm{w}_i(t)\Big(\displaystyle{\sum_{i \in s_e}} \mathrm{w}_i(t)\Big)^{-1}.\label{eq:beta_lossevent}
    \end{align}
\end{proposition}
\begin{proof}
    Averaging equation~\eqref{eq:window_update_individual} over all $i \in s_e$ yields an equation for the evolution of the congestion window $w_{e}(t)$, equating which with equation~\eqref{eq:avg_cwnd_lossevent} yields the above expression for $\beta_e$. 
\end{proof}
Note that during each loss event there may be TCP sources that do not reduce their congestion window, \emph{i.e.}, for each $e\in\{1,2,\cdots,k\}\, \exists\, i \notin s_e$. Therefore, $\beta_e$ obtained for a particular loss event $e$ cannot be used as an estimate of the parameter that governs the multiplicative decrease in the average congestion window $w(t)$. To model the impact of these TCP sources, one may look at the evolution of the average of $w_e(t)$, across loss events. This is analogous to the time average of the average congestion window $w(t) = (\sum_{i=1}^N w_i(t))/N$ with the set of indices in the summation being non-uniform across time. The TCP fluid model relies on ergodicity of the stochastic process that governs the congestion window evolution~\cite{misra2000fluid,raina2005buffer}; hence we invoke the property of the time average of the process being equal to the ensemble average. The following proposition uses this property to derive the required estimate.

\begin{proposition}
At a congested bottleneck link, the multiplicative decrease in the average congestion window $w(t)$ is governed by a parameter $\hat{\beta}$ given by
\begin{align}
    \hat\beta = \displaystyle{\sum_{e=1}^k}\,\displaystyle{\sum_{i\in s_e}} \beta_i \mathrm{w}_i(t)\Big(\displaystyle{\sum_{e=1}^k}\,\displaystyle{\sum_{i\in s_e}} \mathrm{w}_i(t)\Big)^{-1}.\label{eq:beta_estimate_first}
\end{align}
\end{proposition}
\begin{proof}
    When the link is congested, the average congestion window can be approximated as 
    \begin{align}
        w(t) = \frac{1}{k}\displaystyle{\sum_{e=1}^k} w_e(t).\label{eq:ensembleavg_cwnd}
    \end{align}
    From equation~\eqref{eq:avg_cwnd_betaestimate}, the required estimate $\hat\beta$ is given by 
    \begin{align*}
        \hat\beta = 1 - \frac{w(t+1)}{w(t)},
    \end{align*}
    which upon substituting~\eqref{eq:ensembleavg_cwnd} and~\eqref{eq:avg_cwnd_lossevent} yields
    \begin{align}
        \hat\beta = \displaystyle{\sum_{e=1}^k}\, \beta_e w_e(t) \Big(\displaystyle{\sum_{e=1}^k}w_e(t)\Big)^{-1}.\label{eq:beta_estimate_intermediate}
    \end{align}
    As most of the TCP flows would be synchronised due to congestion, one may assume that $N-|s_e| << N\,\,\,\forall\,\, e \in \{1,2,\cdots,k\}$, where $|\cdot|$ represents set cardinality. This leads to the approximation $|s_e| \approx N\,\,\forall\,\, e \in \{1,2,\cdots,k\}$. Now, substituting~\eqref{eq:avg_cwnd_lossevent} and~\eqref{eq:beta_lossevent} in equation~\eqref{eq:beta_estimate_intermediate} yields equation~\eqref{eq:beta_estimate_first}. 
\end{proof}
Observe that the estimate $\hat{\beta}$ given by equation~\eqref{eq:beta_estimate_first} is a weighted average of the multiplicative decrease parameter of the individual TCP sources, \emph{i.e.}, $\beta_i$. Here, each weight is dependent on the summation of the individual congestion window across all loss events. For a congested link, the estimate given by equation~\eqref{eq:beta_estimate_first} can be replaced with a computationally lighter one, as outlined in the proposition below. 
\begin{proposition}
    The estimate $\hat\beta$ given by equation~\eqref{eq:beta_estimate_first} can be replaced by 
    \begin{align}
        \hat{\beta} = \displaystyle{\sum_{i=1}^N}\beta_i\, m_i\Big(\displaystyle{\sum_{i=1}^N m_i}\Big)^{-1},\label{eq:beta_estimate_light}
    \end{align}
    where $m_i$ is the number of times a TCP source $i$ infers packet loss and reduces its congestion window.
\end{proposition}
\begin{proof}
   Equation~\eqref{eq:beta_estimate_first} can be rewritten as 
   \begin{align}
       \hat\beta = \displaystyle{\sum_{e=1}^k}\,\displaystyle{\sum_{i=1}^N} \beta_i \mathrm{w}_i(t)\mathbf{I}(i,e)\Big(\displaystyle{\sum_{e=1}^k}\,\displaystyle{\sum_{i=1}^N} \mathrm{w}_i(t)\mathbf{I}(i,e)\Big)^{-1},\label{eq:beta_estimate_indicatorfunction} 
   \end{align}
   where 
   \begin{align*}
       \mathbf{I}(i,e) = \begin{cases}
           1, & \text{if } i \in s_e\\
           0, & \text{otherwise}.
       \end{cases}
   \end{align*}
   Rearranging the summations in equation~\eqref{eq:beta_estimate_indicatorfunction} yields
     \begin{align}
       \hat\beta = \displaystyle{\sum_{i=1}^N}\, \beta_i \,\displaystyle{\sum_{e=1}^k}\mathrm{w}_i(t)\,\mathbf{I}(i,e)\Big(\displaystyle{\sum_{i=1}^N} \displaystyle{\sum_{e=1}^k}\mathrm{w}_i(t)\,\mathbf{I}(i,e)\Big)^{-1}, \label{eq:beta_estimate_indicatorfunction_2}
   \end{align}
   When the link is congested, a TCP source experiences a packet loss when its sending rate reaches its share of the bottleneck bandwidth. In other words, this happens when the congestion window of the flow has reached its share of the bandwidth-delay product. It is possible that in a single loss event, a given flow may get a smaller share of the bandwidth compared to another. This implies that the first flow may have a smaller congestion window than the second. However, when we observe a given TCP flow $i$ over a period of time and note that it experiences $m_i$ loss events, the summation of the congestion window across all loss events would be approximately equal to $m_i \tilde{C}\tau.$ To that end, we may use the approximation $\sum_{e=1}^k\mathrm{w}_i(t)\,\mathbf{I}(i,e) \approx m_i \tilde{C}\tau.$ Applying this approximation to equation~\eqref{eq:beta_estimate_indicatorfunction_2} yields the revised estimate given by equation~\eqref{eq:beta_estimate_light}.
\end{proof}

    The estimate given by equation~\eqref{eq:beta_estimate_first} requires the SDN controller to maintain records of the congestion window size of all the flows, along with the information of when each flow inferred a packet loss and reduced the congestion window.  Meanwhile, the estimate given by equation~\eqref{eq:beta_estimate_light} requires the router to track only the number of times that a TCP flow reduced the size of its congestion window. The latter appears to require less memory and less computation time. However, either of the two equations,~\eqref{eq:beta_estimate_first} or~\eqref{eq:beta_estimate_light}, may be used to compute the estimate $\hat{\beta}$.

\subsection{Computation of packet-dropping threshold}
With the estimate $\hat{\beta}$ at hand, we may now compute the packet-dropping threshold $q_{th}$ using inequality~\eqref{eq:sufficient_condition}. Note that this condition also requires $w^\ast$, $\tilde{C}$ and $\tau$. Among these, $\tilde{C}$ is the per-flow capacity of the output link of the packet-forwarding device, which is known to the SDN controller, and the average RTT of the flows $\tau$ can be observed by the SDN controller using ack-clocking. Suppose that post congestion detection, the estimate $\hat{\beta}$ is obtained at a time instant $t$, the steady-state average congestion window $w^\ast$ is calculated from the estimated congestion window traces $\mathrm{w}_i(\cdot)$ over a duration $\mathcal{T}$ as 
\begin{align*}
    w^\ast = \frac{1}{\mathcal{T} N}\sum_{l=t-\mathcal{T}}^t \sum_{i=1}^N \mathrm{w}_i(l),
\end{align*}
where $\mathcal{T}$ is the time duration parameter used in Section~\ref{sec:CongestionDetection}.

Given that $q_{th}$ is always a positive integer and that it is always less than the buffer size of the router $B$; the required value of $q_{th}$ can be obtained by brute force in linear time. We begin by computing the function on the left-hand side of inequality~\eqref{eq:sufficient_condition} using an initial value of $q_{th}$, taken as 10; and then recompute this function by incrementing $q_{th}$ until inequality~\eqref{eq:sufficient_condition} is satisfied. These details, along with some checks, are presented in Algorithm~\ref{alg:cap}. 

The algorithm uses one parameter $\kappa$ which serves as an upper bound for the ratio of the average congestion window to the share of bandwidth-delay product of each flow. The physical interpretation of the parameter $\kappa$ emerges from the congestion-control behaviour of AIMD-based TCP itself. Under such congestion control, the congestion window increases additively until a loss is detected, and is reduced by a multiplicative factor thereafter. It is well known that the peak value of the congestion window of a source would ideally be equal to its share of the bandwidth-delay product. However, when packet-forwarding devices use large buffers, often equal to multiples of the bandwidth-delay product as seen in current devices, packets are not dropped adequately; and the congestion window of an end host increases beyond its share of the bandwidth-delay product leading to excessive buffering. This can be checked by the condition $w^\ast \geq \tilde{C}\tau$. When an AQM is to be deployed in such a scenario, the packet-dropping threshold must account for such excessive buffering and drop packets in a manner that it serves as adequate feedback for the end hosts to reduce their congestion window. Therefore, when excessive buffering has occurred we use the parameter $\kappa < 1$ in place of the ratio $w^\ast/(\tilde{C}\tau)$, which is $> 1$, in the function on the left-hand side of inequality~\eqref{eq:sufficient_condition}. This leads to a threshold $q_{th}$ which would drop packets aggressively and bring down the $w^\ast$ appropriately so as to reduce queueing delay. Here, we choose $\kappa =0.9$ heuristically. In general, small values of $\kappa$ would further limit the congestion window of the end hosts, which may in turn affect the throughput observed by the host. 


\begin{algorithm}
\caption{Computation of packet-dropping threshold}\label{alg:cap}
\begin{algorithmic}
\State \textbf{Input:} $\hat{\beta}, w^\ast, \tilde{C}, \tau, B$ 
\State $q_{th} \gets 10$
\State $\kappa \gets 0.9$\\
 \Comment Initialization of $f$
\If{$w^\ast \geq \tilde{C}\tau$} \Comment Checks for excessive buffering
\State $f \gets q_{th}\,\,\hat{\beta}\,\,  \tilde{C}\tau \,\,\displaystyle{\kappa^{q_{th}+1}}$
\Else
\State $f \gets q_{th}\,\,\hat{\beta}\,\,  w^\ast \,\,\displaystyle{\Big(\frac{w^\ast}{\tilde{C}\tau}\Big)^{q_{th}}}$
\EndIf

\vspace{2mm}
\While{$f \geq \pi/2\hspace{2mm} \& \hspace{2mm} q_{th} < B$} \Comment Runs when~\eqref{eq:sufficient_condition} is not true
    \State $q_{th} \gets q_{th} + 1$
    \If{$w^\ast \geq \tilde{C}\tau$}  \Comment Checks for excessive buffering
\State $f \gets q_{th}\,\,\hat{\beta}\,\,  \tilde{C}\tau \,\,\displaystyle{\kappa^{q_{th}+1}}$
\Else
\State $f \gets q_{th}\,\,\hat{\beta}\,\,  w^\ast \,\,\displaystyle{\Big(\frac{w^\ast}{\tilde{C}\tau}\Big)^{q_{th}}}$
\EndIf
\EndWhile \\
\Return $q_{th}$ \Comment Outputs $q_{th}$ and exits
\end{algorithmic}
\end{algorithm}

\section{Performance Evaluation}
\label{sec:PerformanceEvaluation}

We now conduct packet-level simulations on ns-3~\cite{riley2010ns}, version 3.36.1, to evaluate the performance of the threshold-based queue policy, with the threshold $q_{th}$ tuned using the mechanism proposed in Section~\ref{sec:ProposedMechanism}. We consider two network topologies and two traffic compositions in each topology. Simulations are designed with appropriately chosen factors and the performance is benchmarked against two queue policies implemented in today's switches. These details are presented in the following subsections. The entire codebase is available at~\cite{SomitGitHub}.

\subsection{Network topology}
We consider two topologies: the single-bottleneck dumbbell topology, which is a generic setup for AQM evaluation~\cite{rfc7928}, and the more-evolved parking-lot topology, which is used in many congestion-control studies~\cite{katabi2002congestion,wan2025acc,fathalli2025network}. The single-bottleneck topology is consistent with the assumptions of the fluid model~\eqref{eq:avg_cwnd_betaestimate}, which forms the mathematical basis of the proposed mechanism, and is therefore well suited for evaluating its performance. In contrast, the parking-lot topology violates these assumptions and allows us to assess the proposed mechanism in a setting where the fluid model may not apply.

 \emph{Single-bottleneck dumbbell}: We consider $N = 60$ pairs of sending and receiving hosts, connected through a single-bottleneck link. The link has a capacity of $100$ Mbps, and is served by two routers $R1, R2$ on either side. Although we use the term ``router'' for the nodes at either ends of the link, these nodes do not perform any routing, they simply forward packets on the output port leading to a single available path for each TCP flow. These nodes may represent any packet-forwarding device. Each host is connected to the respective router through a dedicated access link of $2$ Mbps capacity. When the access links are fully utilized, the incoming traffic at the router on the sending end is 20\% more than the service capacity of the link, creating congestion.    
 The segment size is fixed as $1500$ bytes. The buffer size $B$ at the router is fixed as the bandwidth-delay product. For the given bandwidth and a round-trip time of $250$ ms, which is generally used for buffer sizing~\cite{raina2005buffer}, we obtain the buffer size as $2084$ pkts.  

 \emph{Parking-lot}: This topology has two bottleneck links of $100$ Mbps capacity in series, with three routers $R1, R2$ and $R3$ serving them. We consider $90$ pairs of sending and receiving hosts. Of the $90$ sending hosts, 60 are connected to $R1$ and $30$ are connected to $R2$. Similarly, of the 90 receiving hosts, 60 are connected to $R3$ and 30 are connected to $R2$. Out of the $60$ hosts connected to $R1$, $30$ hosts send traffic targetted for $30$ hosts connected to $R2$, and the remaining send traffic to $30$ hosts connected to $R3$. Further, the $30$ sending hosts connected to $R2$ send traffic to $30$ of the $60$ receiving hosts connected to $R3$. Consequently, the first bottleneck link is shared by $60$ sending hosts, and traffic pertaining to $30$ of them gets diverted at $R2$, while $30$ new hosts are added, along the network path, at $R2$. Each bottleneck link is thus shared by $60$ sender-receiver pairs. All hosts are connected to their respective routers through a dedicated 2 Mbps access link. A schematic representation of the topology is presented in Fig.~\ref{fig:parking_lot_topology}.

\usetikzlibrary{arrows.meta, decorations.markings, positioning}

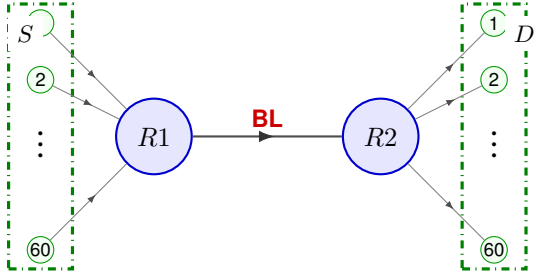
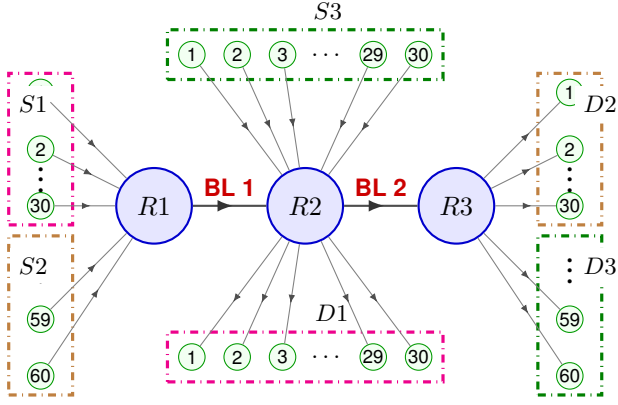
\begin{figure}[tbh!]
\centering
\begin{subfigure}{\linewidth}
\centering
\begin{tikzpicture}[
    router/.style={circle, draw=blue!80!black, fill=blue!10, thick, inner sep=0pt, minimum size=1.0cm, font=\bfseries},
    mininode/.style={circle, draw=green!60!black, fill=green!5, minimum size=0.35cm, inner sep=0pt, font=\scriptsize\sffamily},
    dotsnode/.style={circle, minimum size=0.35cm, inner sep=0pt, font=\small\bfseries},
    directed link/.style={
        thick, 
        draw=black!70,
        postaction={decorate, decoration={
            markings,
            mark=at position 0.55 with {\arrow[black!80]{Latex[scale=0.9]}}
        }}
    },
    host link/.style={
        thin,
        draw=black!50,
        postaction={decorate, decoration={
            markings,
            mark=at position 0.6 with {\arrow[black!70]{Latex[scale=0.7]}}
        }}
    },
    bottleneck/.style={midway, above, font= \small\sffamily\bfseries, color=red!80!black},
    group label/.style={font=\small\bfseries, color=black!90},
    titlerect1/.style={draw,inner ysep=0.5ex,inner xsep=0.5ex,  label={[fill=white,yshift=-0.65cm,xshift=-0.2cm]above:#1}},
    titlerect3/.style={draw,inner ysep=0.5ex,inner xsep=0.5ex,  label={[fill=white,yshift=-0.65cm,xshift=0.4cm]above:#1}},
]

    \node[router] (R1) at (0, 0) {$R1$};
    \node[router] (R3) at (3, 0) {$R2$};


    \node[mininode] (s12-1)  at (-1.5,  1.5) {1};
    \node[mininode] (s12-2)  at (-1.5, 0.75) {2};
    \node[dotsnode] (s12-d1)  at (-1.5,  0) {$\vdots$};
    \node[mininode] (s12-60) at (-1.5, -1.5) {60};
    \node[titlerect1={\small{$S$}}, draw,inner ysep=2pt, inner xsep=4pt, dash dot,fit={(s12-1)  (s12-2) (s12-d1) (s12-60)},color=darkgreen,very thick] (S) {};
    
    \foreach \h in {1,2,60} {
        \draw[host link] (s12-\h) -- (R1);
    }

    \node[mininode] (d23-1)  at (4.5,  1.5) {1};
    \node[mininode] (d23-2)  at (4.5,  0.75) {2};
    \node[dotsnode] (d23-d1)  at (4.5,  0) {$\vdots$};
    \node[mininode] (d23-60) at (4.5, -1.5) {60};
    \node[titlerect3={\small{$D$}}, draw,inner ysep=2pt, inner xsep=4pt, dash dot,fit={(d23-1)(d23-2) (d23-d1) (d23-60)},color=darkgreen,very thick] (D) {};
    
    \foreach \h in {1,2,60} {
        \draw[host link] (R3) -- (d23-\h);
    }

    \draw[directed link] (R1) -- (R3) node[bottleneck] {BL};

\end{tikzpicture}
\caption{\textbf{Single-bottleneck dumbbell.} Traffic flows between source-destination pair across the bottleneck link.}
\end{subfigure}
\begin{subfigure}{\linewidth}
\centering
\begin{tikzpicture}[
    router/.style={circle, draw=blue!80!black, fill=blue!10, thick, inner sep=0pt, minimum size=1.0cm, font=\bfseries},
    mininode/.style={circle, draw=green!60!black, fill=green!5, minimum size=0.35cm, inner sep=0pt, font=\scriptsize\sffamily},
    dotsnode/.style={circle, minimum size=0.35cm, inner sep=0pt, font=\small\bfseries},
    directed link/.style={
        thick, 
        draw=black!70,
        postaction={decorate, decoration={
            markings,
            mark=at position 0.55 with {\arrow[black!80]{Latex[scale=0.9]}}
        }}
    },
    host link/.style={
        thin,
        draw=black!50,
        postaction={decorate, decoration={
            markings,
            mark=at position 0.6 with {\arrow[black!70]{Latex[scale=0.7]}}
        }}
    },
    bottleneck/.style={midway, above, font= \small\sffamily\bfseries, color=red!80!black},
    group label/.style={font=\small\bfseries, color=black!90},
    titlerect1/.style={draw,inner ysep=0.5ex,inner xsep=0.5ex, label={[fill=white,yshift=-0.65cm,xshift=-0.1cm]above:#1}},
    titlerect2/.style={draw,inner ysep=0.5ex,inner xsep=0.5ex, label={[fill=white,xshift=0.3cm]above:#1}},
    titlerect3/.style={draw,inner ysep=0.5ex,inner xsep=0.5ex, label={[fill=white,yshift=-0.65cm,xshift=0.4cm]above:#1}},
    titlerect4/.style={draw,inner ysep=0.5ex,inner xsep=0.5ex, label={[fill=white,xshift=0.35cm]above:#1}}
]

    \node[router] (R1) at (0, 0) {$R1$};
    \node[router] (R2) at (2, 0) {$R2$};
    \node[router] (R3) at (4, 0) {$R3$};

    
    \node[mininode] (s12-1)  at (-1.5,  1.5) {1};
    \node[mininode] (s12-2)  at (-1.5, 0.75) {2};
    \node[dotsnode] (s12-d1)  at (-1.5,  0.45) {$\vdots$};
    \node[mininode] (s12-3)  at (-1.5,  0) {30};
    \node[titlerect1={\small{$S1$}}, draw,inner ysep=2pt, inner xsep=4pt, dash dot,fit={(s12-1) (s12-2) (s12-d1) (s12-3)},color=magenta,very thick] (S1) {};
    \node[dotsnode] (s12-d2)  at (-1.5,  -0.75) {$\vdots$};
    \node[mininode] (s12-59) at (-1.5, -1.5) {59};
    \node[mininode] (s12-60) at (-1.5, -2.25) {60};
    \node[titlerect1={\small{$S2$}}, draw,inner ysep=2pt, inner xsep=4pt, dash dot,fit={(s12-d2) (s12-59) (s12-60)},color=brown,very thick] (S2) {};
    
    \foreach \h in {1,2,3,59,60} {
        \draw[host link] (s12-\h) -- (R1);
    }

    \node[mininode] (s3-1)  at (0.5, 2) {1};
    \node[mininode] (s3-2)  at (1.1, 2) {2};
    \node[mininode] (s3-3)  at (1.7, 2) {3};
    \node[dotsnode] (s3-d)  at (2.3, 2) {$\dots$};
    \node[mininode] (s3-29) at (2.9, 2) {29};
    \node[mininode] (s3-30) at (3.5, 2) {30};
    \node[titlerect2={\small{$S3$}}, draw,inner ysep=3pt, inner xsep=4pt, dash dot,fit={(s3-1) (s3-2) (s3-3) (s3-d) (s3-29) (s3-30)},color=darkgreen,very thick] (S3) {};
    
    \foreach \h in {1,2,3,29,30} {
        \draw[host link] (s3-\h) -- (R2);
    }

    \node[mininode] (d1-1)  at (0.5, -2) {1};
    \node[mininode] (d1-2)  at (1.1, -2) {2};
    \node[mininode] (d1-3)  at (1.7, -2) {3};
    \node[dotsnode] (d1-d)  at (2.3, -2) {$\dots$};
    \node[mininode] (d1-29) at (2.9, -2) {29};
    \node[mininode] (d1-30) at (3.5, -2) {30};
    \node[titlerect4={\small{$D1$}}, draw,inner ysep=3pt, inner xsep=4pt, dash dot,fit={(d1-1) (d1-2) (d1-3) (d1-d) (d1-29) (d1-30)},color=magenta, very thick] (D1) {};
    
    \foreach \h in {1,2,3,29,30} {
        \draw[host link] (R2) -- (d1-\h);
    }

    \node[mininode] (d23-1)  at (5.5,  1.5) {1};
    \node[mininode] (d23-2)  at (5.5,  0.75) {2};
    \node[dotsnode] (d23-d1)  at (5.5,  0.45) {$\vdots$};
    \node[mininode] (d23-3)  at (5.5, 0) {30};
    \node[titlerect3={\small{$D2$}}, draw,inner ysep=2pt, inner xsep=4pt, dash dot,fit={(d23-1)(d23-2)(d23-d1) (d23-3)},color=brown, very thick] (D2) {};
    \node[dotsnode] (d23-d2)  at (5.5,  -0.75) {$\vdots$};
    \node[mininode] (d23-59) at (5.5,  -1.5) {59};
    \node[mininode] (d23-60) at (5.5, -2.25) {60};
    \node[titlerect3={\small{$D3$}}, draw,inner ysep=2pt, inner xsep=4pt, dash dot,fit={(d23-d2)(d23-59)(d23-60)},color=darkgreen,very thick] (D3) {};
    
    \foreach \h in {1,2,3,59,60} {
        \draw[host link] (R3) -- (d23-\h);
    }

    \draw[directed link] (R1) -- (R2) node[bottleneck] {BL 1};
    \draw[directed link] (R2) -- (R3) node[bottleneck] {BL 2};

\end{tikzpicture}
\caption{\textbf{Parking-lot.} Traffic flows from $S1$ to $D1$, $S2$ to $D2$ and from $S3$ to $D3$, each set has 30 hosts. Each of the bottleneck links, BL1 and BL2, carries traffic between $60$ sending-receiving host pairs.}
\end{subfigure}
    \caption{Schematic representation of the two topologies used in simulation. Each access link has a capacity of 2 Mbps, and each bottleneck link is of 100 Mbps.}
    \label{fig:parking_lot_topology}
\end{figure}

\subsection{Traffic composition}
For each topology, we considered two traffic compositions. Similar to the choice of topologies, we choose one traffic composition that aligns with the assumptions of the fluid model, and one that serves as a systematic deviation from these assumptions and is hence closer to the real-world scenario. 
 \begin{figure}[t]
	\centering
	\begin{tikzpicture}
		\begin{axis}[
		yticklabel style={
			/pgf/number format/fixed,
			/pgf/number format/precision=2,
			/pgf/number format/fixed zerofill,
			/pgf/number format/fixed relative,
		},
		every axis plot/.append style={thick},
        every tick label/.append style={font=\normalsize},
        every axis label/.append style={font=\normalsize},
        every axis y label/.append style={yshift=-3pt},
		ylabel = {Parameter estimate $\hat\beta$},
		xlabel = {RTT (ms)},
		xmin = 150, xmax = 300,
		ymin = 0, ymax = 1, 
		grid = both,
		minor tick num = 1,
		major grid style = {lightgray},
		minor grid style = {lightgray!25},
        width=0.5\linewidth,
        height=2in,
	      legend cell align = {right},
		legend pos = south east,
        legend style={draw=none,font=\small}
		]
         \addplot[blue,dash dot,mark=*]  table [col sep=comma,x = T, y = BA] {Data/ba_r0.csv};
         \addplot[darkgreen,solid,mark=square*] table[col sep=comma, x=T, y=BP]  {Data/bp_r0.csv};
		\end{axis}
		\end{tikzpicture}
	    	\begin{tikzpicture}
		\begin{axis}[
		yticklabel style={
			/pgf/number format/fixed,
			/pgf/number format/precision=2,
			/pgf/number format/fixed zerofill,
			/pgf/number format/fixed relative,
		},
		every axis plot/.append style={thick},
        every tick label/.append style={font=\normalsize},
        every axis label/.append style={font=\normalsize},
        every axis y label/.append style={yshift=-3pt},
		ylabel = {Computed threshold $q_{th}$ (pkts)},
		xlabel = {RTT (ms)},
		xmin = 150, xmax = 300,
		  ymin = 10, ymax = 50,
		grid = both,
		minor tick num = 1,
		major grid style = {lightgray},
		minor grid style = {lightgray!25},
        width=0.5\linewidth,
        height=2in,
		legend cell align = {right},
		legend pos = south east
		]
         \addplot[blue,dash dot,mark=*]  table [col sep=comma,x = T, y = QA] {Data/qa_r0.csv};
         \addplot[darkgreen,solid,mark=square*] table[col sep=comma, x=T, y=QP]  {Data/qp_r0.csv};
		\end{axis}
		\end{tikzpicture}
        \begin{tikzpicture}
        \begin{axis}[
            hide axis,
            xmin=0, xmax=1,
            ymin=0, ymax=1,
            legend columns=1,
            legend style={
                draw=none,
                font=\small
            }
        ]
        \addlegendimage{blue,dash dot,mark=*}
        \addlegendentry{obtained using actual \texttt{cwnd}}
        \addlegendimage{darkgreen,solid,mark=square*}
        \addlegendentry{obtained using predicted \texttt{cwnd}}
        \end{axis}
        \end{tikzpicture}
        \caption{A comparison of the estimated parameter $\hat\beta$ and the computed packet-dropping threshold $q_{th}$ when using the congestion window predicted by the Random Forest regressor with those obtained using the actual congestion windows maintained by the source hosts. The common legend is placed below the plots. Observe that the two curves for $q_{th}$ remain close throughout the range of RTT considered.}
        \label{fig:RFR_efficiency}
\end{figure}
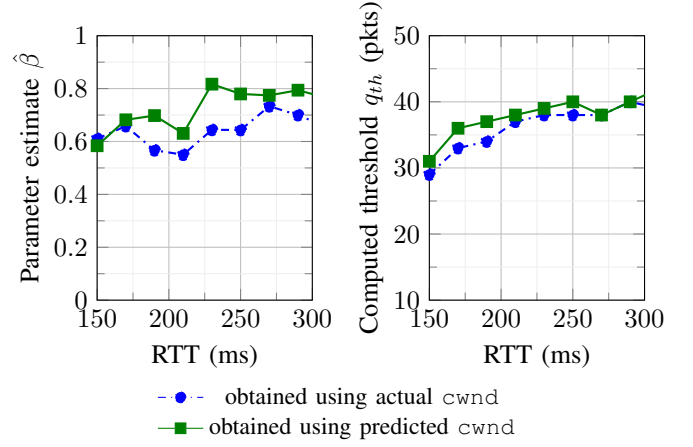

\emph{Homogeneous traffic}: All 60 sending nodes run FTP applications, leading to long-lived TCP flows. For long lived TCP flows \texttt{BulkSendApplication} was installed at each sending host, and a \texttt{PacketSink} was installed at each receiving host. The applications are started at the beginning of each simulation at a time instant chosen randomly from $\mathcal{U}[0,5]$ s. Once started, a long-lived flow remains throughout the simulation time. All the packets have a size of 1500 Bytes.

\emph{Heterogeneous traffic}: In this case, $70$\% of the sources run FTP applications over long-lived TCP flows, $20$\% of them run UDP, while $10$\% of the sending hosts run HTTP applications that generate short-lived TCP flows. For long-lived TCP flows \texttt{BulkSendApplication} was used, as before. For UDP flows, the \texttt{OnOffApplication} was installed at each sending host with \texttt{OffTime} attribute set to $0$, and the \texttt{UdpSocketFactory} was used to ensure that the transport-layer protocol used is UDP. The long-lived TCP flows and UDP flows start at a time instant $\sim\mathcal{U}[0,5]$ s. Meanwhile, the short-lived TCP flows start at a time instant chosen from an exponential distribution with a mean  0.1 s, and is bounded above by 1 s. The duration of each short-lived flow is chosen in a similar way. The time gap between consecutive short-lived flows from a given host is chosen randomly from $\mathcal{U}[0.1,10]$ s. The short-lived flows are also configured using the \texttt{BulkSendApplication} with the start and stops configured using random numbers chosen as described above. Appropriate \texttt{PacketSink} applications are used in all the receiving hosts. All the packets have a size of 1500 Bytes.


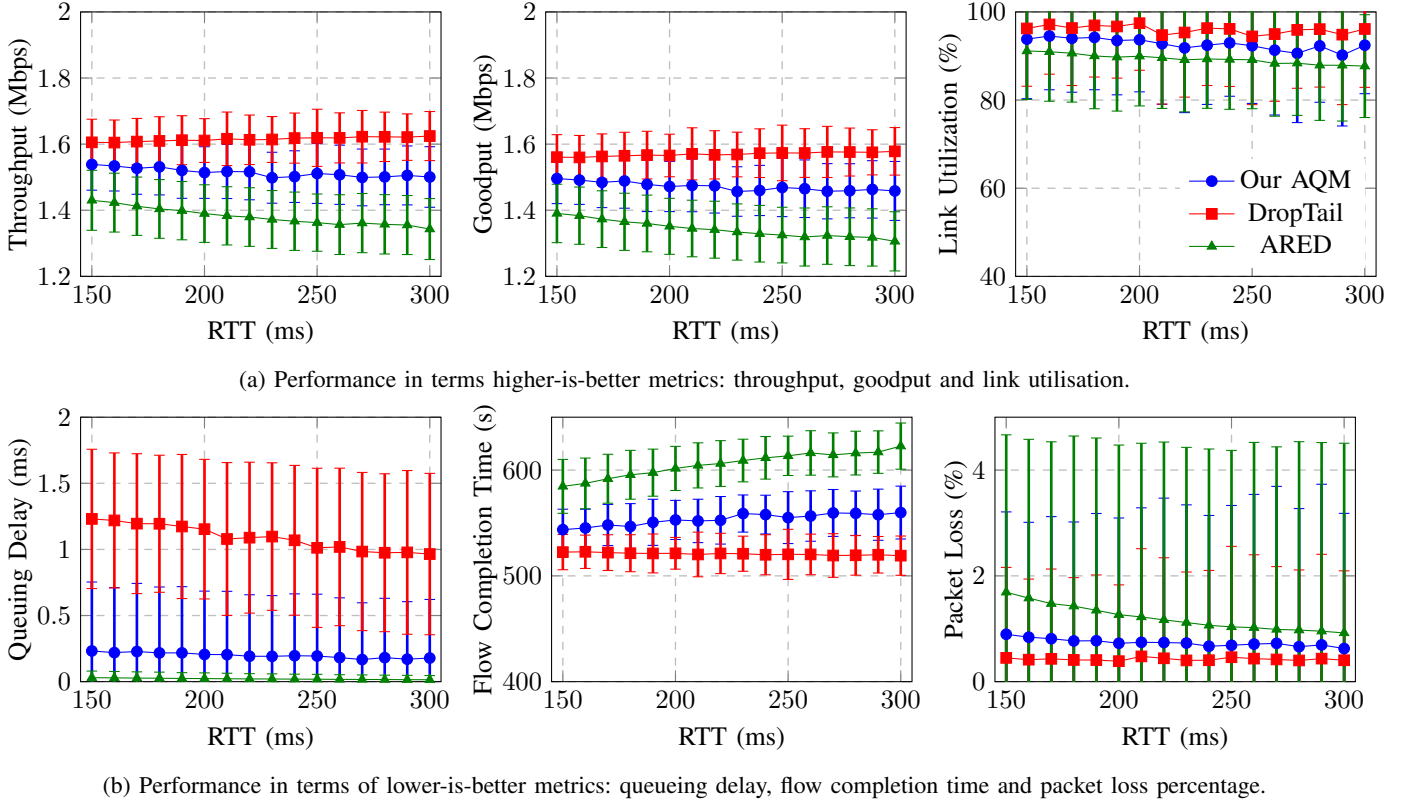
\begin{figure*}[tbh!]
    \centering
    \begin{subfigure}{\textwidth}
        \begin{tikzpicture}
        \def\ycol{avgThroughput}
        \def\yerrcol{stdThroughput}
        \def\ptitle{Average Throughput vs RTT}
        \def\xlabel{RTT}
        \def\ylabel{Throughput (Mbps)}
        \def\ymin{1.2}
        \def\ymax{2}
        \begin{axis}[
            xlabel={RTT (ms)},
            ylabel={\ylabel},
            ylabel shift=-5pt,
            xmin=145, xmax=305,
            ymin=\ymin, ymax=\ymax,
            xtick={150, 200, 250, 300},
            legend pos=south west,
            grid=both,
            major grid style={dashed,line width=0.5pt,draw=gray!50},
            minor grid style={dotted,line width=0.25pt,draw=gray!25},
            yticklabel style={/pgf/number format/fixed},
            width=2.5in,
            height=2in,
        ]
        
        \addplot[
            color=blue,
            mark=*,
            mark options={solid},
            error bars/.cd,
            y dir=both,
            y explicit,
            error bar style={line width=1pt},
        ] table[x=RTT,
        y=\ycol, 
        y error=\yerrcol,
        col sep=comma,each nth point={2}] {Data/results_LN_aqm_zc_150_to_300_TM_mlBeta.csv};
        
        \addplot[
            color=red,
            mark=square*,
            mark options={solid},
            error bars/.cd,
            y dir=both,
            y explicit,
            error bar style={line width=1pt},
        ] table[x=RTT,
        y=\ycol, 
        y error=\yerrcol,
        col sep=comma,each nth point={2}] {Data/results_LN_naqm_150_to_300_TM_mlBeta.csv};
        
        \addplot[
            color=darkgreen,
            mark=triangle*,
            mark options={solid},
            error bars/.cd,
            y dir=both,
            y explicit,
            error bar style={line width=1pt},
        ] table[x=RTT,
        y=\ycol,
        y error=\yerrcol,
        col sep=comma,each nth point={2}] {Data/results_LN_red_150_to_300_TM_mlBeta.csv};
        
        \end{axis}
        \end{tikzpicture} 
        \begin{tikzpicture}
        \def\ycol{avgGoodput}
        \def\yerrcol{stdGoodput}
        \def\ptitle{Average Goodput vs RTT}
        \def\xlabel{RTT}
        \def\ylabel{Goodput (Mbps)}
        \def\ymin{1.2}
        \def\ymax{2}
        \begin{axis}[
            xlabel={RTT (ms)},
            ylabel={\ylabel},
            ylabel shift=-5pt,
            xmin=145, xmax=305,
            ymin=\ymin, ymax=\ymax,
            xtick={150, 200, 250, 300},
            legend pos=south west,
            grid=both,
            major grid style={dashed,line width=0.5pt,draw=gray!50},
            minor grid style={dotted,line width=0.25pt,draw=gray!25},
            yticklabel style={/pgf/number format/fixed},
            width=2.5in,
            height=2in,
        ]
        
        \addplot[
            color=blue,
            mark=*,
            mark options={solid},
            error bars/.cd,
            y dir=both,
            y explicit,
            error bar style={line width=1pt},
        ] table[x=RTT,
        y=\ycol, 
        y error=\yerrcol,
        col sep=comma,each nth point={2}] {Data/results_LN_aqm_zc_150_to_300_TM_mlBeta.csv};
        
        \addplot[
            color=red,
            mark=square*,
            mark options={solid},
            error bars/.cd,
            y dir=both,
            y explicit,
            error bar style={line width=1pt},
        ] table[x=RTT,
        y=\ycol, 
        y error=\yerrcol,
        col sep=comma,each nth point={2}] {Data/results_LN_naqm_150_to_300_TM_mlBeta.csv};
        
        \addplot[
            color=darkgreen,
            mark=triangle*,
            mark options={solid},
            error bars/.cd,
            y dir=both,
            y explicit,
            error bar style={line width=1pt},
        ] table[x=RTT,
        y=\ycol,
        y error=\yerrcol,
        col sep=comma,each nth point={2}] {Data/results_LN_red_150_to_300_TM_mlBeta.csv};
        \end{axis}
        \end{tikzpicture}
        \begin{tikzpicture}
            \def\ycol{linkUtilization}
            \def\yerrcol{stdLinkUtilization}
            \def\ptitle{Link Utilization (\%) vs RTT}
            \def\xlabel{RTT}
            \def\ylabel{Link Utilization (\%)}
            \def\ymin{40}
            \def\ymax{100}
            \begin{axis}[
                xlabel={RTT (ms)},
                ylabel={\ylabel},
                ylabel shift=-5pt,
                xmin=145, xmax=305,
                ymin=\ymin, ymax=\ymax,
                xtick={150, 200, 250, 300},
                legend pos=south east,
                legend style={draw=none},
                grid=both,
                major grid style={dashed,line width=0.5pt,draw=gray!50},
                minor grid style={dotted,line width=0.25pt,draw=gray!25},
                yticklabel style={/pgf/number format/fixed},
                width=2.5in,
                height=2in,
            ]
            
            \addplot[
                color=blue,
                mark=*,
                mark options={solid},
                error bars/.cd,
                y dir=both,
                y explicit,
                error bar style={line width=1pt},
            ] table[x=RTT,
            y=\ycol, 
            y error=\yerrcol,
            col sep=comma,each nth point={2}] {Data/results_LN_aqm_zc_150_to_300_TM_mlBeta.csv};
            \addlegendentry{Our AQM}
            
            \addplot[
                color=red,
                mark=square*,
                mark options={solid},
                error bars/.cd,
                y dir=both,
                y explicit,
                error bar style={line width=1pt},
            ] table[x=RTT,
            y=\ycol, 
            y error=\yerrcol,
            col sep=comma,each nth point={2}] {Data/results_LN_naqm_150_to_300_TM_mlBeta.csv};
            \addlegendentry{DropTail}
            
            \addplot[
                color=darkgreen,
                mark=triangle*,
                mark options={solid},
                error bars/.cd,
                y dir=both,
                y explicit,
                error bar style={line width=1pt},
            ] table[x=RTT,
            y=\ycol,
            y error=\yerrcol,
            col sep=comma,each nth point={2}] {Data/results_LN_red_150_to_300_TM_mlBeta.csv};
            \addlegendentry{ARED}
            
            \end{axis}
        \end{tikzpicture}
    \caption{Performance in terms higher-is-better metrics: throughput, goodput and link utilisation. }
    \end{subfigure}\\
    \begin{subfigure}{\textwidth}
        \begin{tikzpicture}
         \def\ycol{queuingDelay}
        \def\yerrcol{stdQueuingDelay}
        \def\ptitle{ Queuing Delay vs RTT}
        \def\xlabel{RTT}
        \def\ylabel{Queuing Delay (ms)}
        \def\ymin{0}
        \def\ymax{2}
            \begin{axis}[
                xlabel={RTT (ms)},
                ylabel={\ylabel},
                ylabel shift=-5pt,
                xmin=145, xmax=305,
                ymin=\ymin, ymax=\ymax,
                xtick={150, 200, 250, 300},
                legend pos=north east,
                grid=both,
                major grid style={dashed,line width=0.5pt,draw=gray!50},
                minor grid style={dotted,line width=0.25pt,draw=gray!25},
                yticklabel style={/pgf/number format/fixed},
                width=2.5in,
                height=2in,
            ]
            
            \addplot[
                color=blue,
                mark=*,
                mark options={solid},
                error bars/.cd,
                y dir=both,
                y explicit,
                error bar style={line width=1pt},
            ] table[x=RTT,
            y=\ycol, 
            y error=\yerrcol,
            col sep=comma,each nth point={2}] {Data/results_LN_aqm_zc_150_to_300_TM_mlBeta.csv};
            
            \addplot[
                color=red,
                mark=square*,
                mark options={solid},
                error bars/.cd,
                y dir=both,
                y explicit,
                error bar style={line width=1pt},
            ] table[x=RTT,
            y=\ycol, 
            y error=\yerrcol,
            col sep=comma,each nth point={2}] {Data/results_LN_naqm_150_to_300_TM_mlBeta.csv};
            
            \addplot[
                color=darkgreen,
                mark=triangle*,
                mark options={solid},
                error bars/.cd,
                y dir=both,
                y explicit,
                error bar style={line width=1pt},
            ] table[x=RTT,
            y=\ycol,
            y error=\yerrcol,
            col sep=comma,each nth point={2}] {Data/results_LN_red_150_to_300_TM_mlBeta.csv};
            
            \end{axis}
        \end{tikzpicture}
        \begin{tikzpicture}
        \def\ycol{flowCompTime}
        \def\yerrcol{stdFlowCompTime}
        \def\ptitle{Flow Completion Time vs RTT}
        \def\xlabel{RTT}
        \def\ylabel{Flow Completion Time (s)}
        \def\ymin{400}
        \def\ymax{650}
            \begin{axis}[
                xlabel={RTT (ms)},
                ylabel={\ylabel},
                ylabel shift=-5pt,
                xmin=145, xmax=305,
                ymin=\ymin, ymax=\ymax,
                xtick={150, 200, 250, 300},
                legend pos=south east,
                grid=both,
                major grid style={dashed,line width=0.5pt,draw=gray!50},
                minor grid style={dotted,line width=0.25pt,draw=gray!25},
                yticklabel style={/pgf/number format/fixed},
                width=2.5in,
                height=2in,
            ]
            
            \addplot[
                color=blue,
                mark=*,
                mark options={solid},
                error bars/.cd,
                y dir=both,
                y explicit,
                error bar style={line width=1pt},
            ] table[x=RTT,
            y=\ycol, 
            y error=\yerrcol,
            col sep=comma,each nth point={2}] {Data/results_LN_aqm_zc_150_to_300_TM_mlBeta.csv};
            
            \addplot[
                color=red,
                mark=square*,
                mark options={solid},
                error bars/.cd,
                y dir=both,
                y explicit,
                error bar style={line width=1pt},
            ] table[x=RTT,
            y=\ycol, 
            y error=\yerrcol,
            col sep=comma,each nth point={2}] {Data/results_LN_naqm_150_to_300_TM_mlBeta.csv};
            
            \addplot[
                color=darkgreen,
                mark=triangle*,
                mark options={solid},
                error bars/.cd,
                y dir=both,
                y explicit,
                error bar style={line width=1pt},
            ] table[x=RTT,
            y=\ycol,
            y error=\yerrcol,
            col sep=comma,each nth point={2}] {Data/results_LN_red_150_to_300_TM_mlBeta.csv};
            
            \end{axis}
    \end{tikzpicture}
    \begin{tikzpicture}
        \def\ycol{packetLoss}
        \def\yerrcol{stdPktLoss}
        \def\ptitle{ Packet Loss (\%) vs RTT}
        \def\xlabel{RTT}
        \def\ylabel{Packet Loss (\%)}
        \def\ymin{0}
        \def\ymax{5}
        \begin{axis}[
                xlabel={RTT (ms)},
                ylabel={\ylabel},
                ylabel shift=-5pt,
                xmin=145, xmax=305,
                ymin=\ymin, ymax=\ymax,
                xtick={150, 200, 250, 300},
                legend pos=north east,
                legend style={draw=none},
                grid=both,
                major grid style={dashed,line width=0.5pt,draw=gray!50},
                minor grid style={dotted,line width=0.25pt,draw=gray!25},
                yticklabel style={/pgf/number format/fixed},
                width=2.5in,
                height=2in,
            ]
            
            \addplot[
                color=blue,
                mark=*,
                mark options={solid},
                error bars/.cd,
                y dir=both,
                y explicit,
                error bar style={line width=1pt},
            ] table[x=RTT,
            y=\ycol, 
            y error=\yerrcol,
            col sep=comma,each nth point={2}] {Data/results_LN_aqm_zc_150_to_300_TM_mlBeta.csv};
            
            \addplot[
                color=red,
                mark=square*,
                mark options={solid},
                error bars/.cd,
                y dir=both,
                y explicit,
                error bar style={line width=1pt},
            ] table[x=RTT,
            y=\ycol, 
            y error=\yerrcol,
            col sep=comma,each nth point={2}] {Data/results_LN_naqm_150_to_300_TM_mlBeta.csv};
            
            \addplot[
                color=darkgreen,
                mark=triangle*,
                mark options={solid},
                error bars/.cd,
                y dir=both,
                y explicit,
                error bar style={line width=1pt},
            ] table[x=RTT,
            y=\ycol,
            y error=\yerrcol,
            col sep=comma,each nth point={2}] {Data/results_LN_red_150_to_300_TM_mlBeta.csv};
        
        \end{axis}
    \end{tikzpicture}
    \caption{Performance in terms of lower-is-better metrics: queueing delay, flow completion time and packet loss percentage.}   
    \end{subfigure}
    \caption{\textbf{Single-bottleneck dumbbell, heterogeneous traffic}. Plots showing the performance of the threshold-based queue policy, with the packet-dropping threshold tuned using the proposed mechanism, benchmarked against the DropTail and ARED policies. The legend placed in the top-right plot applies to all the plots.}
    \label{fig:perf_eval_results_singlebottleneck}
\end{figure*}

\subsection{Threshold policy with proposed mechanism}
For queue management, at $R1$ in the single-bottleneck toplogy and $R1$ and $R2$ in the parking-lot topology, we used the threshold-based policy along with the four components of the proposed mechanism: congestion detection, congestion window prediction, parameter estimation, and threshold computation, implemented as separate functions emulating the SDN control server (see Fig.~\ref{fig:schematic_implementation}). For congestion detection, we used Algorithm~\ref{alg:CongestionDetection}. For this, the parameter $m_z$ was set as 5, and the time window $\mathcal{T}$ was fixed as 5~s, a few multiples of the largest RTT used in the simulations. As observed through the ROC plots presented in Section~\ref{sec:CongestionDetection}, Algorithm~\ref{alg:CongestionDetection} detects congestion successfully for these, and regardless of these, values. 
The congestion window estimation runs in parallel and uses the sequence numbers of the packets along with the IP addresses and port numbers of the sender-receiver pairs to estimate the required $\mathrm{w}_i(t)$ waveforms for all the sending hosts. Once congestion is detected, these $\mathrm{w}_i(t)$ waveforms are fed into beta estimation function, which estimates the multiplicative decrease parameter using equation~\eqref{eq:beta_estimate_light}. The resultant estimate $\hat{\beta}$, along with (i) $w^\ast$ obtained from the $\mathrm{w}_i(t)$ waveforms for all hosts, (ii) the average RTT estimated using the sequence numbers, and (iii) the capacity of the outgoing link, is fed to Algorithm~\ref{alg:cap} which returns the packet-dropping threshold $q_{th}$. This threshold is then used to make packet-dropping decisions at the router, thereby deploying the AQM. A schematic representation of the implementation is shown in Fig.~\ref{fig:schematic_PM}.  

\subsection{Benchmarking, variable factors and metrics}
We benchmark the performance of the threshold-based policy with the proposed mechanism against two policies that are implemented in today's routers/switches. The first one is the simple DropTail policy which simply drops packets once the buffer is full. Secondly, we consider the Adaptive Random Early Detection (ARED) which is implemented in most Linux-based packet-forwarding devices. For the ARED implementation, we used the parameters $minQth = 50$ pkts, $maxQth = 100$ pkts which are comparable with the $q_{th}$ computed for the threshold-based policy. Although numerous other AQMs are proposed in the literature, both classical as well as ML-aided, we choose these two for a comparative study given their implementation status. If an AQM is to be deployed in live-networks, it must outperform those that are already implemented and are hence closest to deployment. 

Along with varying the topology and traffic composition, as described above, we also consider two variable factors: the round-trip time, and the TCP variant. Round-trip time is a crucial factor as it is an uncontrollable network parameter; in that, longer paths with large propagation delay could cause large RTTs regardless of network conditions. Given that operating regimes with large RTTs are known to exhibit congestion in networks carrying TCP flows~\cite{raina2005buffer,ghosh2020right}, we vary RTT in the range $[150, 300]$ ms. In order to avoid RTT-induced synchronization among hosts, in each simulation with a fixed RTT of  $\tau \in [150, 300]$ ms, the RTT for each sender-receiver pair is chosen randomly as $T \sim \mathcal{U}[\tau-5,\tau+5]$ ms. To fix the required RTT, each access link is configured with a one-way delay of $T/4$, which results in a mean RTT of $\tau$ across flows. Two AIMD-based TCP variants, namely \texttt{LinuxReno} and \texttt{NewReno} are considered. 

We consider six performance metrics: (i) throughput, (ii) goodput, (iii) link utilisation, (iv) packet-loss percentage, (v) queueing delay, and (vi) flow-completion time. Of these, the first three are higher-is-better metrics, and the remaining are lower-is-better. The choice of metrics is consistent with the recommendations in RFC 7928~\cite{rfc7928} for evaluation of AQMs. 

For recording all metrics, each simulation is run until all clients have transmitted $100$ MB Data. 
Among these metrics, throughput and goodput were recorded at the receiving hosts. Here, throughput considers all packet transfers while goodput measures the rate of transfer of useful packets, \emph{i.e.}, retransmissions do not contribute to goodput. Note that these two are user-oriented metrics. Link utilisation, queueing delay, and packet loss were recorded at the bottleneck link. These serve as system-oriented metrics. 
Flow completion time is defined as the time required to complete the transfer of a file of a fixed size, and is often regarded as the most important user-oriented metric. 
For this, we fix a file size of $100$ MB at the FTP application of each host and run the simulation until the last flow was completed. 

\input{Sections_J/Results_Plots_parkingLot_TM}

\subsection{Results and discussion}

Firstly, to verify the efficiency of the proposed congestion window estimation using the Random Forest regressor, we simulated the two succeeding components, \emph{i.e.} $\beta$ estimation and $q_{th}$ computation, under two settings: (i) using the actual congestion window (\texttt{cwnd}) traces obtained using the \texttt{CongestionWindow} trace source of the socket at each TCP source, (ii) using the congestion window waveform estimated by the Random Forest. Here, the first setting serves as a baseline to compare the second setting with. The parameter estimate $\hat\beta$ and the computed packet-dropping threshold $q_{th}$ obtained for the two cases are plotted in Fig.~\ref{fig:RFR_efficiency}. From these, we observe that using the congestion window estimates produced by the Random Forest, as prescribed by the proposed mechanism, does not cause any significant deviation in the threshold $q_{th}$ across the range of RTT considered. Notably, the curves get closer as the RTT increases, indicating that the window estimation is particularly effective in the large RTT regime which is prone to performance degradation in general. 

Next, we evaluate the performance of the proposed method using the required metrics in the RTT range considered. In each simulation, throughput, goodput and packet loss are recorded at the receiving hosts using \texttt{FlowMonitor}.  Link utilisation and queueing delay are recorded at the router preceding the bottleneck link, through the appropriate calculations on the readouts of the \texttt{PointToPointNetDevice/PhyTx} and the \texttt{RootQueueDiscList/$\ast$/PacketsInQueue} trace source, respectively. 
From the traces obtained for each simulation, the first few seconds of data corresponding to transient behaviour are discarded and mean and standard deviation of each metric are computed from the remaining data.

For flow completion time, the time required for the 100 MB transfer is recorded for each flow, and the mean and standard deviation across flows is computed. The results for the single bottleneck dumbbell topology with the LinuxReno TCP variant used in heterogeneous traffic composition, for the whole range of RTT are presented in Fig.~\ref{fig:perf_eval_results_singlebottleneck}. For the same setting, the results obtained with the parking-lot topology are presented in Fig.~\ref{fig:perf_eval_results_parking-lot}. In all these plots, we used the dot-whisker representation, where the dot denotes the mean and the upper and lower whiskers are used to indicate the range of values that are within one standard deviation from the mean. We skip the results obtained for homogeneous traffic, and TCP NewReno as they qualitatively appear similar to the ones presented. We now discuss our observations.

\emph{User-oriented performance}: In the single bottleneck network, the mean throughput and goodput offered by our AQM implementation is seen to be marginally below that of DropTail, and better than that of ARED. Notably, the throughput and goodput remain constant across RTT in the case of our AQM, and ARED exhibits a reduction in throughput and goodput as RTT increases. In particular, the goodput performance of our AQM is agnostic to the RTT, even in operating regimes with large RTTs. This is due to the fact that the $q_{th}$ is tuned according to condition~\eqref{eq:sufficient_condition} which guarantees network stability regardless of the RTT. In contrast, the goodput achieved by ARED suffers as RTT increases. Similar observations hold for the case of flow completion time, which is a lower-is-better metric as opposed to the former two. Our AQM has a marginally higher mean flow completion time as compared to DropTail, with the standard deviation of the two indicating that the observed flow completion time for a give flow is likely to be similar for the two policies. Meanwhile, the flow completion time of ARED is larger than that of our AQM, and shows a substantial increase as RTT increases. For the parking-lot topology, the three policies show similar relative performance. Additionally, for this topology, the standard deviation of throughput and goodput of ARED is much larger than that of our AQM and DropTail. This indicates that the hosts would experience pronounced throughput fluctuations during the life of a connection. Likewise, a large standard deviation in flow completion time indicates a large variation in the time taken by hosts to complete the file transfer, and hence is reflective of unfair network behaviour.

\emph{System-oriented performance}: From the plots showing the link utilisation, it can be observed that DropTail guarantees a marginally better utilisation of the link than our AQM, with the gap being less for the parking-lot topology. ARED lacks in this aspect, with the utilisation reducing as RTT increases. Though this is seen in both topologies, it is more pronounced in the parking-lot topology. This is because DropTail does not take any action towards congestion mitigation, it simply drops packets once the buffer is full. In doing so, it keeps the buffer full most of the time, thereby never starving the link, and hence ensuring a higher utilisation. On the other hand, our AQM and ARED are meant to drop packets in a manner that congestion is alleviated. By dropping packets appropriately, they ensure that the queue size remains small, which could sometimes lead to loss of link utilisation. This is indeed reflected in the queueing delay and packet loss measurements. One can observe that our AQM and ARED ensure a much smaller queueing delay compared to DropTail. Although our AQM has a higher queueing delay than ARED, it can be observed that this delay does not reflect in the flow completion time. The observations regarding queueing delay are consistent across the two topologies. 
From the packet loss percentage, it is seen that in the single bottleneck dumbbell topology, DropTail drops the least number of packets while ARED drops the most, our AQM lies between the two. For large RTTs, the fraction of packets dropped by our AQM is very similar to that of DropTail, which indicates that our AQM is able to maintain substantially less queueing delay while dropping just marginally more packets. In the parking-lot topology, we see that our AQM exhibits a lower packet loss compared to DropTail, with the reduction being more evident for large RTTs. ARED shows a larger packet loss for small RTTs, and nearly converges to that of our AQM as RTT increases. Arguably, our AQM ensures better packet loss than DropTail in networks where links appear in tandem. These results show that our AQM drop packets optimally, signalling congestion to the TCP hosts who in turn respond by reducing their sending rates. This is the reason why the threshold-based policy implemented using the proposed mechanism is able to guarantee a performance similar to DropTail in terms of most metrics, while substantially outperforming DropTail in queueing delay.

\section{Conclusions}
\label{sec:Conclusions}
Most Active Queue Management (AQM) strategies aimed at alleviating network congestion have tunable parameters that need to be co-designed with parameters of the transport protocols, such as TCP, implemented in the end hosts. Toward facilitating the deployment of such AQMs, we examined the possibility of estimating TCP parameters through packet traces accessible to the router and co-designing AQM parameters with these. In particular, we considered a threshold-based queue policy which requires the packet-dropping threshold to be co-designed with TCP parameters. This work outlines a mechanism that detects network congestion, predicts the evolution of the congestion window at end hosts, estimates the required TCP parameter, and produces a value of the packet-dropping threshold. Appropriate algorithms are outlined for congestion detection and threshold computation. The mechanism used to estimate the TCP parameter is supported by arguments based on existing results from the analysis of TCP fluid models. 
We conducted packet-level simulations to evaluate the performance of the threshold-based policy deployed as per the proposed mechanism. The results show that our implementation of the AQM performs equivalent to DropTail in terms of some metrics, and particularly outperforms it in terms of the queueing delay. It is also seen to be better than ARED, which is implemented in most Linux-based routers today, in terms of most performance metrics. Although the proposed mechanism is grounded in mathematical models that account for AIMD behaviour of TCP alone, our simulations reveal that it performs well in networks carrying unresponsive UDP flows and bursty short-lived TCP flows as well. 
The proposed mechanism appears to be a promising solution towards TCP-aware queue management through appropriate deployment of AQMs to alleviate congestion in packet-forwarding devices. It is particularly suited for SDN architectures, which provide the flexibility of control- and data-plane separation. 
\subsection{Future work}
This work relies on results from analysis of a TCP fluid model that captures the behaviour of the AIMD congestion control algorithm. However, the proposed mechanism only needs to estimate the multiplicative decrease parameter, regardless of those that govern the additive increase. This raises the question of whether the proposed mechanism would work for traffic comprising TCP flows that employ a congestion control algorithm that decreases the congestion window multiplicatively upon detecting congestion, but increases it using a different approach. For example, CUBIC (default in Linux OS), increases its congestion window in a non-linear manner, but reduces it multiplicatively in the event of congestion. It is desirable to explore whether the proposed mechanism would work for such TCP flows. Further, it is desirable to evaluate the performance of the proposed mechanism in switches that support Explicit Congestion Notification (ECN), this would also pave the way for examining its performance with TCP variants that rely on ECN marks, instead of packet drops, for congestion feedback.

\balance
\bibliographystyle{IEEEtran}  
\bibliography{main} 

\begin{thebibliography}{10}
\providecommand{\url}[1]{#1}
\csname url@samestyle\endcsname
\providecommand{\newblock}{\relax}
\providecommand{\bibinfo}[2]{#2}
\providecommand{\BIBentrySTDinterwordspacing}{\spaceskip=0pt\relax}
\providecommand{\BIBentryALTinterwordstretchfactor}{4}
\providecommand{\BIBentryALTinterwordspacing}{\spaceskip=\fontdimen2\font plus
\BIBentryALTinterwordstretchfactor\fontdimen3\font minus \fontdimen4\font\relax}
\providecommand{\BIBforeignlanguage}[2]{{%
\expandafter\ifx\csname l@#1\endcsname\relax
\typeout{** WARNING: IEEEtran.bst: No hyphenation pattern has been}%
\typeout{** loaded for the language `#1'. Using the pattern for}%
\typeout{** the default language instead.}%
\else
\language=\csname l@#1\endcsname
\fi
#2}}
\providecommand{\BIBdecl}{\relax}
\BIBdecl

\bibitem{McFillin2022}
\BIBentryALTinterwordspacing
A.~McFillin, ``Bufferbloat may be solved, but it’s not over yet,'' accessed July 2025. [Online]. Available: \url{https://blog.apnic.net/2020/01/22/bufferbloat-may-be-solved-but-its-not-over-yet/}
\BIBentrySTDinterwordspacing

\bibitem{gettys2012bufferbloat}
J.~Gettys and K.~Nichols, ``Bufferbloat: dark buffers in the internet,'' \emph{Communications of the ACM}, vol.~55, no.~1, pp. 57--65, 2012.

\bibitem{allman2012comments}
M.~Allman, ``Comments on bufferbloat,'' \emph{ACM SIGCOMM Computer Communication Review}, vol.~43, no.~1, pp. 30--37, 2012.

\bibitem{cerf2014bufferbloat}
V.~G. Cerf, ``Bufferbloat and other {Internet} challenges,'' \emph{IEEE Internet Computing}, vol.~18, no.~5, pp. 80--80, 2014.

\bibitem{nichols2012controlling}
K.~Nichols and V.~Jacobson, ``Controlling queue delay,'' \emph{Communications of the ACM}, vol.~55, no.~7, pp. 42--50, 2012.

\bibitem{harrison2023buffer}
P.~Harrison, ``Buffer-bloated router? {How to prevent it and improve performance},'' \emph{Communications of the ACM}, vol.~66, no.~6, pp. 73--77, 2023.

\bibitem{adams2012active}
R.~Adams, ``Active queue management: A survey,'' \emph{IEEE Communications Surveys \& Tutorials}, vol.~15, no.~3, pp. 1425--1476, 2012.

\bibitem{baker2015ietf}
F.~Baker and G.~Fairhurst, ``{IETF} recommendations regarding active queue management,'' Tech. Rep., 2015.

\bibitem{toopchinezhad2025machine}
M.~P. Toopchinezhad and M.~Ahmadi, ``Machine learning approaches for active queue management: A survey, taxonomy, and future directions,'' \emph{Computer Networks}, p. 111174, 2025.

\bibitem{bouacida2018practical}
N.~Bouacida and B.~Shihada, ``Practical and dynamic buffer sizing using {LearnQueue},'' \emph{IEEE Transactions on Mobile Computing}, vol.~18, no.~8, pp. 1885--1897, 2018.

\bibitem{jung2023rlecn}
E.-S. Jung, H.~S. Kim \emph{et~al.}, ``{RLECN—A} learning based dynamic threshold control of {ECN},'' \emph{ICT Express}, vol.~9, no.~6, pp. 1007--1012, 2023.

\bibitem{cheng2024pet}
K.~Cheng, T.~Wang, X.~Du, S.~Du, and H.~Cai, ``{PET}: Multi-agent independent {PPO}-based automatic {ECN} tuning for high-speed data center networks,'' \emph{arXiv preprint arXiv:2405.11956}, 2024.

\bibitem{jiang2021machine}
H.~Jiang, Q.~Li, Y.~Jiang, G.~Shen, R.~Sinnott, C.~Tian, and M.~Xu, ``When machine learning meets congestion control: A survey and comparison,'' \emph{Computer Networks}, vol. 192, p. 108033, 2021.

\bibitem{kundel2018p4}
R.~Kundel, J.~Blendin, T.~Viernickel, B.~Koldehofe, and R.~Steinmetz, ``{P4-Codel}: Active queue management in programmable data planes,'' in \emph{2018 IEEE Conference on Network Function Virtualization and Software Defined Networks {(NFV-SDN)}}.\hskip 1em plus 0.5em minus 0.4em\relax IEEE, 2018, pp. 1--4.

\bibitem{de2026ired}
L.~C. De~Almeida, P.~D. Maciel, R.~Pasquini, C.~Papagianni, and F.~L. Verdi, ``{iRED}: A disaggregated {P4-AQM} fully implemented in programmable data plane hardware,'' \emph{IEEE Transactions on Networking}, 2026.

\bibitem{de2016pi2}
K.~De~Schepper, O.~Bondarenko, I.-J. Tsang, and B.~Briscoe, ``{PI2: A linearized AQM for both classic and scalable TCP},'' in \emph{Proceedings of the 12th International on Conference on emerging Networking EXperiments and Technologies}, 2016, pp. 105--119.

\bibitem{la2004instability}
R.~J. La, ``Instability of a tandem network and its propagation under {RED},'' \emph{IEEE Transactions on Automatic Control}, vol.~49, no.~6, pp. 1006--1011, 2004.

\bibitem{hollot2002analysis}
C.~V. Hollot, V.~Misra, D.~Towsley, and W.~Gong, ``Analysis and design of controllers for {AQM} routers supporting {TCP} flows,'' \emph{IEEE Transactions on Automatic Control}, vol.~47, no.~6, pp. 945--959, 2002.

\bibitem{manjunath2018stability}
S.~Manjunath and G.~Raina, ``Stability and performance of compound {TCP} with a proportional integral queue policy,'' \emph{IEEE Transactions on Control Systems Technology}, vol.~27, no.~5, pp. 2139--2155, 2018.

\bibitem{manjunath2019compound}
------, ``{Compound TCP with Random Early Detection (RED)}: stability, bifurcation and performance analyses,'' \emph{arXiv preprint arXiv:1907.06302}, 2019.

\bibitem{raina2015stability}
G.~Raina, S.~Manjunath, S.~Prasad, and K.~Giridhar, ``Stability and performance analysis of compound {TCP} with {REM} and drop-tail queue management,'' \emph{IEEE/ACM Transactions on Networking}, vol.~24, no.~4, pp. 1961--1974, 2015.

\bibitem{riley2010ns}
G.~F. Riley and T.~R. Henderson, ``The {NS}-3 network simulator,'' in \emph{Modeling and tools for network simulation}.\hskip 1em plus 0.5em minus 0.4em\relax Springer, 2010, pp. 15--34.

\bibitem{bosshart2014p4}
P.~Bosshart, D.~Daly, G.~Gibb, M.~Izzard, N.~McKeown, J.~Rexford, C.~Schlesinger, D.~Talayco, A.~Vahdat, G.~Varghese \emph{et~al.}, ``{P4}: Programming protocol-independent packet processors,'' \emph{ACM SIGCOMM Computer Communication Review}, vol.~44, no.~3, pp. 87--95, 2014.

\bibitem{misra2000fluid}
V.~Misra, W.-B. Gong, and D.~Towsley, ``Fluid-based analysis of a network of {AQM} routers supporting {TCP} flows with an application to red,'' in \emph{Proceedings of the conference on Applications, Technologies, Architectures, and Protocols for Computer Communication}, 2000, pp. 151--160.

\bibitem{raina2005part}
G.~Raina, D.~Towsley, and D.~Wischik, ``{Part II: Control theory for buffer sizing},'' \emph{ACM SIGCOMM Computer Communication Review}, vol.~35, no.~3, pp. 79--82, 2005.

\bibitem{han2005synchronization}
H.~Han, C.~V. Hollot, D.~Towsley, and Y.~Chait, ``Synchronization of {TCP} flows in networks with small droptail buffers,'' in \emph{Proceedings of the 44th IEEE Conference on Decision and Control}.\hskip 1em plus 0.5em minus 0.4em\relax IEEE, 2005, pp. 6762--6767.

\bibitem{ghosh2020right}
D.~Ghosh, K.~Jagannathan, and G.~Raina, ``Right buffer sizing matters: Some dynamical and statistical studies on compound {TCP},'' \emph{Performance Evaluation}, vol. 139, p. 102095, 2020.

\bibitem{raina2005buffer}
G.~Raina and D.~Wischik, ``Buffer sizes for large multiplexers: {TCP} queueing theory and instability analysis,'' in \emph{Next Generation Internet Networks, 2005}.\hskip 1em plus 0.5em minus 0.4em\relax IEEE, 2005, pp. 173--180.

\bibitem{SomitGitHub}
S.~Gond, ``Repository name: {MTP},'' \url{https://github.com/somitgond/mtp}, 2026.

\bibitem{rfc7928}
\BIBentryALTinterwordspacing
N.~Kuhn, P.~Natarajan, N.~Khademi \emph{et~al.}, ``{Characterization Guidelines for Active Queue Management (AQM)},'' IETF, Tech. Rep. RFC 7928, Jul. 2016. [Online]. Available: \url{https://www.rfc-editor.org/info/rfc7928}
\BIBentrySTDinterwordspacing

\bibitem{katabi2002congestion}
D.~Katabi, M.~Handley, and C.~Rohrs, ``Congestion control for high bandwidth-delay product networks,'' in \emph{Proceedings of the 2002 conference on Applications, technologies, architectures, and protocols for computer communications}, 2002, pp. 89--102.

\bibitem{wan2025acc}
Z.~Wan, J.~Zhang, S.~Cheng, Y.~Chen, T.~Pan, P.~Lin, and T.~Huang, ``{ACC}: addressing performance limitations in datacenters with atomic congestion control,'' in \emph{IEEE INFOCOM 2025-IEEE Conference on Computer Communications}.\hskip 1em plus 0.5em minus 0.4em\relax IEEE, 2025, pp. 1--10.

\bibitem{fathalli2025network}
S.~Fathalli, E.~N. Weyulu, D.~Zeynali, B.~Chandrasekaran, and A.~Feldmann, ``Network-assisted congestion feedback,'' \emph{IEEE Transactions on Network and Service Management}, vol.~23, pp. 1797--1815, 2025.

\end{thebibliography}

\end{document}